\documentclass[journal,10pt,twocolumn]{IEEEtranTCOM}
\usepackage{algorithm,amsbsy,amsmath,amssymb,epsfig,bbm,mathrsfs,multirow,amsthm}
\usepackage{color}
\usepackage{algpseudocode}
\usepackage{graphicx,caption,subcaption,array,multirow}
\usepackage{makecell}
\usepackage{mathtools}
\usepackage[hidelinks]{hyperref}

\DeclarePairedDelimiter\floor{\lfloor}{\rfloor}

\DeclareMathAlphabet{\mathpzc}{OT1}{pzc}{m}{it}
\DeclareMathOperator*{\argmax}{\arg\!\max}

\newtheorem{lemma}{Lemma}
\newtheorem{theorem}{\textbf{\textsc{Theorem}}}

\begin{document}
	
\title{Optimal Beam Association for High Mobility mmWave Vehicular Networks: Lightweight Parallel Reinforcement Learning Approach}

\author{
	Nguyen Van Huynh, Diep N. Nguyen, Dinh Thai Hoang, and Eryk Dutkiewicz
\thanks{The authors are with the School of Electrical and Data Engineering, University of Technology Sydney, Sydney, NSW 2007, Australia. E-mails: huynh.nguyenvan@student.uts.edu.au, \{Diep.Nguyen, Hoang.Dinh, and Eryk.Dutkiewicz\}@uts.edu.au.}
\thanks{Preliminary results in this paper were presented at the IEEE GLOBECOM Conference, 2020~\cite{Globecom}.}
}
\maketitle
\thispagestyle{empty}
\begin{abstract}
In intelligent transportation systems (ITS), vehicles are expected to feature with advanced applications and services which demand ultra-high data rates and low-latency communications. For that, the millimeter wave (mmWave) communication has been emerging as a very promising solution. However, incorporating the mmWave into ITS is particularly challenging due to the high mobility of vehicles and the inherent sensitivity of mmWave beams to dynamic blockages. This article addresses these problems by developing an optimal beam association framework for mmWave vehicular networks under high mobility. Specifically, we use the semi-Markov decision process to capture the dynamics and uncertainty of the environment. The Q-learning algorithm is then often used to find the optimal policy. However, Q-learning is notorious for its slow-convergence. Instead of adopting deep reinforcement learning structures (like most works in the literature), we leverage the fact that there are usually multiple vehicles on the road to speed up the learning process. To that end, we develop a lightweight yet very effective parallel Q-learning algorithm to quickly obtain the optimal policy by simultaneously learning from various vehicles. Extensive simulations demonstrate that our proposed solution can increase the data rate by 47\% and reduce the disconnection probability by 29\% compared to other solutions.
\end{abstract}


\section{Introduction}
Over the past few years, the explosive growth of interest in intelligent transportation and vehicular communications offers a great potential to enhance traffic efficiency, improve road safety, and enable open disruptive entertainment services and autonomous driving~\cite{Va2016Millimeter},~\cite{Yan2019Channel}. These applications often require low-latency, high reliability, and especially multi-Gbps network access. For instance, Google's self-driving car in a second can generate up to $750$ MB of data~\cite{AngelicaGoogle}. It is expected that a vehicle may produce $1$ Terabyte of data in a single trip~\cite{Va2016Millimeter}. To address this critical problem, the emerging millimeter wave (mmWave) communication has been recently considered as a very promising solution~\cite{Wang2018MillimeterSurvey}. Comparing with existing wireless networks, the mmWave technology operates at much higher carrier frequencies, i.e., from 30 GHz to 300 GHz. Thus, it possesses much more abundant spectrum resources, resulting in potentially extremely high data rates and low-latency communications. Nevertheless, in mmWave communications, the temporal degradation of channel quality occurs much more frequently than conventional (lower frequency) communications due to high propagation attenuation, selective directivity, and severe susceptibility to blockages, especially in vehicular communications under high mobility. This work aims to address these problems to enable mmWave communications under high mobility.

\subsection{Related Work and Motivation}
Various solutions in the literature have been proposed to address the above inherent limitations of mmWave communications, e.g., \cite{Mismar2020Deep}-\cite{Sohrabi2017Hybrid}. In~\cite{Mismar2020Deep}, the authors aimed to jointly optimize the transmit power, interference coordination, and beamforming to maximize the signal to interference plus noise ratio and the sum-rate. In~\cite{Semiari2017Joint}, the authors aimed to enhance the users' quality-of-experience by scheduling communication tasks on both microwave and mmWave bands. Similarly, a new protocol that enables simultaneous connections to conventional 4G cells and 5G mmWave cells is introduced in~\cite{Polese2017Improved}. Differently, a new adaptive beamforming strategy is proposed in~\cite{Chiu2019Active} to improve the system communication rate by taking the noise into account and selecting the beamforming vectors based on the posterior of the angle-of-arrival.

The aforementioned solutions and most others in the literature did not account for the high mobility of vehicles and dynamics of the environment. In~\cite{Garcia2016Location}, the authors proposed an adaptive channel estimation mechanism for beamforming with the aid of location information. The position information is then demonstrated as an important factor to greatly improve the initial association of vehicles to the infrastructure. Similarly, in~\cite{Va2016Beam}, the authors introduced a beam switching mechanism in mmWave vehicle-to-infrastructure communications. In~\cite{Alrabeiah2020Deep}, a deep learning approach is proposed to predict mmWave beams and blockages with sub-6 GHz channels. Specifically, the authors studied that there are deterministic mapping from uplink sub-6 GHz channels to downlink mmWave channels. Hence, with the uplink channel knowledge, it is possible to predict/infer mmWave beams and blockages with low training overhead because sub-6 GHz channels can be estimated by using a small number of pilots. Nevertheless, it is very challenging to derive and analyze these mapping functions. To tackle this, the authors developed a deep learning model with a sufficiently large deep neural network to efficiently predict the optimal beams and blockages status.

To deal with the inherent dynamic and uncertainty of mmWave links' quality, the authors in~\cite{Sim2018An} proposed an online learning algorithm to obtain the optimal beam selection policy based on the prior environment information. This problem is first modeled as a contextual multi-armed bandit problem. Then, the learning algorithm is developed to guide the mmWave base station to select an optimal subset of beams for vehicles. In~\cite{Scalabrin2018Beam}, the authors aimed to maximize the number of bits delivered from a base station to a mobile user by considering the trade-off between directive data transmission (DT) and directional beam training (BT) in mmWave vehicular networks. Specifically, the system is formulated as a partially observable Markov decision process with the system state defined as the position of the mobile user within the road link. Note that in both~\cite{Sim2018An} and~\cite{Scalabrin2018Beam}, only one mmWave base station is considered. In practice where multiple base stations are often in place, they can cooperate in the beam association and handoff process to achieve a globally optimal beam selection solution. An online learning approach is proposed in~\cite{Hussain2020Adaptive} to jointly optimize beam training, data transmission, and handover processes. Nevertheless, the effect of high mobility on the system performance has not been studied. Moreover, similar to~\cite{Scalabrin2018Beam}, only one mobile user is considered in this paper. In~\cite{Khan2019Reinforcement}, the authors considered the vehicle-cell association problem for mmWave vehicular networks to maximize the average rate of vehicles. Specifically, the authors first formulated the association problem as a discrete non-convex optimization problem. Then, a learning algorithm is developed to estimate the solution for the non-convex optimization problem. It is worth noting that the above works and others in the literature have not accounted for the effects of the high mobility and blockage on beam association/handover. In addition, these learning approaches do not leverage the fact that there are multiple vehicles running on the road at the same time to improve the learning efficiency.

Given the above, this work aims to develop an optimal beam association framework for mmWave vehicular communications under the high mobility of vehicles and the uncertainty of blockages. In particular, to capture the dynamics of blockages, channel quality, and mobility, we first model the problem as a semi-Markov decision process (SMDP). The Q-learning (QL) algorithm is then often adopted to solve the optimization problem the underlying SMDP. Nevertheless, the QL algorithm is well known for its slow convergence rate, especially in dynamic and complicated environments. Instead of adopting deep reinforcement learning structures e.g., double deep QL, deep dueling, in this article, we leverage the fact that there are usually multiple vehicles on the road to speed up the convergence to the optimal solution. To that end, we develop a lightweight yet very effective parallel QL algorithm to quickly obtain the optimal policy by simultaneously learning from various vehicles. Specifically, vehicles on the road can act as active learners to help the system simultaneously collect data. Based on the collected data, the proposed parallel QL algorithm can \emph{quickly learn} the environment information, e.g., RSSI profile, beam's location, and blockage, to derive the optimal beam association strategy. The proposed parallel QL algorithm does not only require lower complexity but also converge faster than the latest deep learning-based approaches (e.g., double QL, deep dueling). Moreover, unlike deep reinforcement learning methods (e.g., \cite{Khan2019Reinforcement}), our proposed parallel QL framework is proven to always converge to the optimal policy. We show that the high mobility and parallelism of vehicles now become helpful in speeding up the learning process of our underlying algorithm.

Moreover, unlike existing works, e.g., \cite{Sim2018An},~\cite{Khan2019Reinforcement}, in which learning algorithms are deployed at the vehicles with limited resources, in our design, the eNodeB executes the parallel QL algorithm and sends optimal beam association/handover actions to vehicles. As such, the computing complexity is moved to the eNodeB. The simulation results then show that our proposed solution can increase the data rate for each vehicle by up to 47\% and reduce the disconnection probability by 29\% compared to existing approaches.

\subsection{Main Contributions}
The major contributions of this paper are as follows.

\begin{itemize}
	\item Develop an optimal beam association solution for high-mobility mmWave vehicular communication networks using the semi-Markov decision process framework that can effectively deal with the inherent dynamics of the mmWave connection quality and the uncertainty of the environment, e.g., beam's location, RSSI profile, the velocity of the vehicle, and blockages, in a real-time manner.
	
	\item Develop a lightweight yet very effective parallel QL algorithm to quickly obtain the optimal policy by simultaneously learning from various vehicles to update the global Q-table at the eNodeB. The proposed parallel QL algorithm does not only require lower complexity but also converge faster than advanced deep reinforcement learning approaches (e.g., deep double and deep dueling). Specifically, the algorithm deploys multiple learning processes at the eNodeB, and each learning process is assigned to learn from a vehicle on the considered road to update the global policy.
	
	\item Prove that the proposed parallel QL framework converges with probability one to the optimal policy. Note that deep QL based approaches (e.g., deep double and deep dueling) with advanced deep neural networks are not always guaranteed to converge. We also provide a comprehensive analysis of the convergence time/rate, complexity, and overhead of the proposed framework.
	
	\item Perform extensive simulations to demonstrate the effectiveness of the proposed parallel QL algorithm. By learning from multiple vehicles and exploiting the high mobility of vehicles, our proposed algorithm can achieve the performance close to that of the hypothetical scheme which requires complete environment information in advance.
\end{itemize}
	
The rest of paper is organized as follows. The system model is described in Section~\ref{Sec.System}. Section~\ref{Sec:prob} presents the problem formulation based on the SMDP. Then, Section~\ref{sec:Q} describes the conventional QL algorithm and our proposed parallel QL algorithm. After that, we provide the evaluation results in Section~\ref{sec:evaluation}. Finally, the conclusion is highlighted in Section~\ref{sec:conclusion}.

\section{System Model}
\label{Sec.System}
\begin{figure}[h]
	\centering
	\includegraphics[scale=0.4]{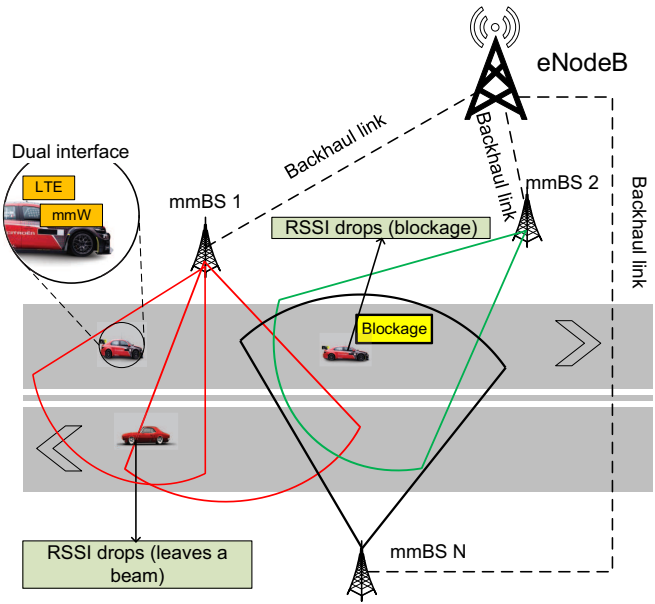}
	\caption{System model.}
	\label{Fig.system_model}
\end{figure}

Consider a millimeter wave (mmWave) vehicular network, where an LTE eNodeB and a set of $N$ mmWave base stations (mmBSs) $BS = \{BS_1, \ldots, BS_n, \ldots, BS_N\}$ are deployed as shown in Fig.~\ref{Fig.system_model}. This is an expected network model for 5G and beyond systems~\cite{Giordani2016Multi},~\cite{Shokri2016Spectrum}. All the mmBSs can connect to the eNodeB via backhaul links. Each vehicle is equipped with two communication interfaces: (i) an LTE interface to communicate with the eNodeB and (ii) an mmWave interface to communicate with an mmBS~\cite{Sim2018An}. The Friss free-space equation reveals that with given transmit power and antenna gains, the pathloss increases when the frequency increases~\cite{Va2016Millimeter},~\cite{Jameel2019Propa}. For example, increasing the carrier frequency from 2 GHz to 60 GHz~\footnote{60 GHz is the common frequency band in automotive communications.} results in an increase of 29 dB for the pathloss. The mmWave signals also suffer from obstacles and scattering objects (both static and dynamic). Thus, the path loss model can be formulated as follows~\cite{Va2016Millimeter},~\cite{Jameel2019Propa}:
\begin{equation}
\label{eq:pathloss}
\begin{aligned}
PL(d)[dB] &= PL(d_0) + 10n\log_{10}\frac{d}{d_0} + \psi\\
& =20\log_{10}\frac{4\pi d_0}{\lambda} + 10n\log_{10}\frac{d}{d_0} + \psi,
\end{aligned}
\end{equation}
where $d_0$ is the reference distance, $PL(d_0)$ is the free-space loss at the reference distance, $n$ is the path loss exponent, $d$ is the distance between the vehicle and the mmBS, $\lambda$ is the carrier wavelength, and $\psi$ is the log-normal \textit{shadowing loss} caused by the absorption of obstacles and scattering objects. Moreover, the log-normal \textit{shadowing loss} increases with distance between the vehicle and the mmBS as the farther the transmitter and the receiver are the more likely that there are obstructing objects in between them. When the path loss increases, the received power at the mmBS decreases. Thus, when the vehicle leaves a beam or moves to a blockage zone, the received signal strength indicator (RSSI) will drop. As such, the vehicle is not be able to communicate with its connected mmBS (through the beam). To avoid disrupting the service, the vehicle needs to connect to another beam which provides better channel quality. This beam can belong to the current connected mmBS, i.e., beam association, or belong to another mmBS, i.e., handover.

Conventional methods in the literature~\cite{Navarro2008MDP}-\cite{Elshaer2016Downlink} usually make a beam association/handover decision based on the current channel information or network state, e.g., SINR or RSSI, where the decision is triggered when the SINR/RSSI are dropped due to blockage objects or mobility. However, these solutions may lead to too frequent handover and hence the associated handover cost/interruption, especially in mmWave networks where the temporal degradation of channel quality frequently occurs due to the intermittent connectivity of mmWave links. In addition, as the mmWave communication quality often deteriorates intermittently and rapidly, these solutions may lead to the ping-pong effect resulting in high outage probability and low system performance~\cite{Zang2019Managing},~\cite{Sun2018Smart}. 

In this paper, our learning algorithm can efficiently address these critical problems by learning the environment conditions. In particular, we consider a centralized controller, i.e., the eNodeB, that ``learns" from vehicles on the considered road and makes beam association/handover decisions for all vehicles. With the proposed learning algorithm, the eNodeB can learn the RSSI profiles of the beams in the systems. To reduce the number of beam association/handovers, the eNodeB can guide the vehicle to connect to a beam with a ``better" RSSI profile (in terms of the long-run average data rate). In addition, the beam association/handover decision can be triggered even when the RSSI level is still good to anticipate the intermittent problem of mmWave links. It is worth noting that the RSSI level can be inferred by the mmBS through the received signals from the vehicle. Moreover, with current standards in ITS systems~\cite{Standard}, the mmBS always has the location of its connected vehicles. As such, the eNodeB runs the algorithm to obtain the best beam for the vehicle to connect without adding noticeable overhead to current ITS systems.

We assume that each mmBS $BS_n$ has a finite set $\mathcal{B}_n = \{b_{n,1}, \ldots, b_{n,k}, \ldots, b_{n,\mathrm{K}}\}$ of $K$ orthogonal beams~\cite{Sim2018An}. Based on the information learned from the vehicle, i.e., location, velocity, and RSSI level, the eNodeB selects beam $b_{n,k}$ of mmBS $BS_n$ to support the communications of the vehicle. In this paper, the velocity of each vehicle is not fixed. When the vehicle is connected to beam $b_{n,k}$, it can successfully transmit data with rate $r_{n,k}$. Note that $r_{n,k}$ is a random variable, depending on the RSSI level (i.e., channel quality) of the channel. Consider $M$ RSSI levels $\mathcal{R}=\{0, 1, \ldots, M-1\}$ which depend on the environmental conditions, e.g., channel conditions (as modeled above) and blockage probability. The higher the RSSI level is, the higher the achievable data rate of the vehicle. We assume that when the vehicle enters blockage zone, the RSSI level drops to $0$, and thus the vehicle cannot connect to the mmBS. We then define $\omega_{b_{n,k}}$ is the blockage probability of beam $b_{n,k}$ with $0 \leq \omega_{b_{n,k}} \leq 1$. $\omega_{b_{n,k}} = 1$ if there are static blockage objects (e.g., buildings) in the coverage of the beam. Note that the blockage model is not the input of our proposed algorithm, i.e., it is not required to be available in advance. Instead, our proposed algorithm can learn the blockage model by interacting with the environment and observing the immediate reward. Thus, our proposed framework can work with any probabilistic blockage model (i.i.d or non i.i.d). Similar to~\cite{Sim2018An}, we assume that $r_{n,k}$ is varied from 0 to $R_{\max}$, where $R_{\max}$ is the maximum achievable rate, corresponding to the highest RSSI level. Formally, $r_{n,k}$ can be formulated as in (\ref{eq:rate}).
\begin{equation}
\label{eq:rate}
r_{n,k}	=	\left\{	\begin{array}{ll}
0,	&	\mbox{with probability $\omega_{b_{n,k}}$},	\\
\Delta(l), l \in \{\mathcal{R} \setminus \{0\}\},	&	\mbox{with probability $1 - \omega_{b_{n,k}}$},
\end{array}	\right.
\end{equation}
where $\Delta(l)$ is the rate corresponding to the current RSSI level $l$~\cite{Sim2018An}. Note that to capture the fading effect of the channel, $\Delta(l)$ follows a given random distribution.

Note that although the above environment information (e.g., the RSSI profiles and the blockage probability) are required for the modeling/formulation purpose, the proposed parallel QL algorithm below does not require these parameters explicitly as input. Instead, after executing an action, the eNodeB observes the \emph{reward}, i.e., the actual data rate between the vehicle and the connected mmBS. The reward function (defined below) hence captures the communication channel between vehicles and mmBSs (e.g., the bit error rate and fading). The details of the proposed algorithm are described in the next section.


\section{Problem Formulation}
\label{Sec:prob}
This work aims to deal with the dynamics of blockages, channel quality, and mobility. However, the conventional MDP is not effective in capturing the dynamics and uncertainty of the system. Thus, we propose to use the SMDP~\cite{Puterman_1994_Book}. Different from the MDP, in an SMDP, an action is only taken when an event occurs. An SMDP can be defined as a tuple $<t_j, \mathcal{S}, \mathcal{A}, r>$, where $\mathcal{S}$ and $\mathcal{A}$ are the state space and the action space of the system, respectively. $t_j$ defines decision epoch $j$-th when an event occurs, and $r$ is the reward function.

\subsection{State Space}
The state space of each vehicle is represented as the discretized space of RSSI levels, the connected beam, the speed, and the direction of the vehicle. Formally, the state space $\mathcal{S}$ is defined as follows:
\begin{equation}
	\begin{aligned}
		\mathcal{S} \triangleq & \Bigl\{ (l, b_{n,k}, v, d): l \in \{0, \ldots, m, \ldots, M-1\},\\
		&b_{n,k} \in \mathcal{B}_n \cup \{b_{0,0}\}, \forall n \in \{1, \ldots, N\}, \forall k \in \{1, \ldots, K\},\\
		& v \in \{0, \ldots, v_\mathrm{max}\}, d \in \{0, 1\}\Bigr\},
	\end{aligned}
\end{equation}
where $l$ is the RSSI level of the current vehicle and $b_{n,k}$ is the current connected beam of the current vehicle. $v$ is the vehicle's speed, $v_\mathrm{max}$ is the maximum speed of the vehicle, and $d$ is the direction of the vehicle (0 and 1 represent two directions of the vehicle). $b_{0,0}$ is a virtual beam used to capture the case that there is no available beam at a given location of the current vehicle. In this case, the RSSI level of the current vehicle is always $0$. Note that due to the inherent dynamics of mmWave link's quality, the mmWave instantaneous RSSI estimate can be less accurate, providing less useful information to decide the next action. This is one of the key motivations of our machine learning-based framework that can learn the RSSI profiles of beams by interacting with the environment and observing the actual reward (transmission rate) other than just relying on the instantaneous RSSI estimate.

The road is modeled as a one-dimension area $\mathcal{W} \in \mathbb{R}$ that is discretized to $\floor*{\frac{W}{z}}$ zones, where $W$ is the length of the considered road, $z$ is the length of each zone, and $\floor{.}: \mathbb{R} \rightarrow \mathbb{N}$ is the floor function. When the vehicle at location $w \in \mathcal{W}$, the vehicle is at zone $\floor*{\frac{w}{z}}$-th. At the current state $s \in \mathcal{S}$, an event $e_s$ is triggered if a vehicle reaches a new zone. Note that as the vehicle's speed is not fixed, the time interval between two consecutive epochs varies. To capture that, the SMDP is used in this work, instead of the conventional Markov decision process with identical time slots. Specifically, the eNodeB only takes an action when event $e_s$ occurs. Note that unlike existing works that have the location information as part of the state space, our framework captures the location information via events $e_s$ under the SMDP formulation. This allows us to efficiently deal with dynamics and uncertainties during the learning process, which is very beneficial in our considered system with the dynamics of blockages, channel quality, and mobility.

\subsection{Action Space}

When the vehicle reaches a new zone on the road, i.e., event $e_s$ is triggered (given its current state $s \in \mathcal{S}$), the eNodeB decides if the vehicle needs to associate to a new beam or stay on the current beam. The action space $\mathcal{A}_s$ is defined as:
\begin{equation}
\mathcal{A}_s \triangleq \{ a\}= \{b_{n,k}, b_{0,0}\}, \forall n, \forall k,
\end{equation}
where $a$ is the action made at state $s$. $a = b_{n,k}$ if the eNodeB guides the vehicle to connect to beam $b_{n,k}$. This includes the case staying with the current beam. $a = b_{0,0}$ if there is no available beam at the current location.

\subsection{Immediate Reward}
\label{subsec:immediate}
In this paper, we aim to maximize the long-term average data rate of the system. As mentioned, at decision epoch $t$, if an action is taken so that the vehicle connects to beam $b_{n,k}$, it can communicate with a rate of $0 \leq r_{n,k} \leq R_{max}$ corresponding to the current RSSI level. The resulting data that the vehicle receives from the mmBS is calculated as $j_{n,k}^t r_{n,k}$, where $j_{n,k}^t$ is the connection time between two consecutive decision epochs (during which the vehicle can communicate with the mmBS through beam $b_{n,k}$). As the algorithm observes the reward at the end of each decision epoch, $j_{n,k}^t$ is the duration from the time that the vehicle enters the current zone until it leaves to enter the next zone. $j_{n,k}^t$ hence depends on the speed of the vehicle $v_t$ at the epoch $t$. Practically, the speed $v_t$ can change from one to another epoch or even during the time $j_{n,k}^t$. However, without loss of generality, we assume that the time $j_{n,k}^t$ between two consecutive epochs is small enough (e.g., by setting the length per zone $z$ as small as necessary) so that the vehicle's RSSI level and speed remain unchanged. Thus, connection time $j_{n,k}^t$ can be calculated as $\frac{z}{v_t}$. 

In addition, the service may be interrupted during the handover/beam-switching, denoted as $h$, i.e., the time it takes for the vehicle to switch to the new mmBS. We assume that the handover time is the same for all the mmBSs. Taking the handover time into account, at state $s \in \mathcal{S}$, the immediate data rate after performing action $a$ is obtained in (\ref{5}).
\begin{equation}\label{5}
r(s_t,a_t)	=	\left\{	\begin{array}{ll}
(j_{n,k}^t - h) r_{n,k},	&	\mbox{if $a_t = b_{n,k}$ and the vehicle}\\
&\mbox{handovers to a new mmBS},	\\
j_{n,k}^t r_{n,k},	&	\mbox{if $a_t = b_{n,k}$ and the vehicle}\\
&\mbox{stays in the current mmBS},	\\
0,	&	\mbox{otherwise},
\end{array}	\right.
\end{equation}
where $s_t$ and $a_t$ are the system state and the action taken at decision epoch $t$, respectively and $r_{n,k}$ is the communication rate when the vehicle connects to beam $b_{n,k}$ as defined in (\ref{eq:rate}). Note that $j_{n,k}^t$, $r_{n,k}$, and $h$ depend on the channel and environment conditions such as blockages, beam's location, and RSSI profiles. Under our design, they are implicitly learned through interacting with the environment and observing the immediate reward in an online manner.

\subsection{Optimization Formulation}

The decision policy $\pi$ of the proposed SMDP can be defined as a mapping from the state space to the action space: $\mathcal{S} \rightarrow \mathcal{A}_s$~\cite{Puterman_1994_Book},~\cite{Li2018SMDP}. Thus, with initial state $s$, the long-term average data rate is formulated as follows:
\begin{equation}
\label{eq:average_reward}
\mathcal{R}_{\pi}(s) = \lim\limits_{T \rightarrow \infty} \frac{\mathbb{E}\{\sum_{t=0}^{T} r(s_t, \pi(s_t)) | s_0 = s\}}{\mathbb{E}\{\sum_{t=0}^{T}\xi_t | s_0 = s\}}, \forall s \in \mathcal{S},
\end{equation}
where $\xi_t$ is the time interval between the $t$-th and $(t+1)$-th decision epochs, $\pi(s)$ is the action at state $s$ based on policy $\pi$, and $r$ is the immediate reward after performing an action. In Theorem~\ref{theo:reward}, we will prove that the limit in (\ref{eq:average_reward}) exits~\cite{JSAC2}.
\begin{theorem}
	\label{theo:reward}
	With the number of events in a given time and the number of states in the state space $\mathcal{S}$ are finite, we have:
	\begin{equation}
	\begin{aligned}
	\label{eq:reward_theorem}
	\mathcal{R}_{\pi}(s) &= \lim\limits_{T \rightarrow \infty} \frac{\mathbb{E}\{\sum_{t=0}^{T} r({s}_t, \pi(s_t)) | s_0 = s\}}{\mathbb{E}\{\sum_{t=0}^{T}\xi_t | s_0 = s\}}\\
	&= \frac{\overline{\mathcal{L}}_\pi r(s,\pi(s))}{\overline{\mathcal{L}}_\pi y(s,\pi(s))}, \forall s \in \mathcal{S},
	\end{aligned}
	\end{equation}
	where $y(s,\pi(s))$ denotes the expected time interval between two consecutive decision epochs when an action is taken at state $s$ following policy $\pi$. $\overline{\mathcal{L}}_\pi$ is the limiting matrix of the transition probability matrix $\mathcal{L}_\pi$~\cite{Puterman_1994_Book}, defined as:
	\begin{equation}
	\overline{\mathcal{L}}_\pi = \lim\limits_{T \rightarrow \infty} \frac{1}{T}\sum_{t=0}^{T-1} \mathcal{L}_\pi^t,
	\end{equation}
\end{theorem}
\begin{proof}
	First, we prove the following lemma.
	\begin{lemma}
		\label{Lem:exist}
		The limiting matrix $\overline{\mathcal{L}}_\pi$ of the the transition probability matrix $\mathcal{L}_\pi$ always exists.
	\end{lemma}
	\begin{proof}
	The proof of Lemma~\ref{Lem:exist} is provided in Appendix~\ref{app:reward}.
	\end{proof}
	As the limiting matrix $\overline{\mathcal{L}}_\pi$ exits (see Lemma~\ref{Lem:exist}) and the total probabilities of transiting from a given state to other states equals to 1, i.e., $\sum_{s'\in \mathcal{S}}^{}\overline{\mathcal{L}}_\pi(s'|s) = 1$ , we have:
	\begin{equation}
	\begin{aligned}
	&\overline{\mathcal{L}}_\pi r(s, \pi(s)) = \lim\limits_{T \rightarrow \infty} \frac{1}{T+1}\mathbb{E} \Bigl\{\sum_{t=0}^{J}r(s_t, \pi(s_t))\Bigr\}, \forall s \in \mathcal{S}, \\
	&\overline{\mathcal{L}}_\pi y(s, \pi(s)) = \lim\limits_{T \rightarrow \infty} \frac{1}{T+1}\mathbb{E}\Bigl\{\sum_{t=0}^{T}\xi_t\Bigr\}, \forall s \in \mathcal{S}.
	\end{aligned}
	\end{equation}
	Clearly, the long-term average reward in (\ref{eq:reward_theorem}) is obtained by taking the ratio of $\overline{\mathcal{L}}_\pi r(s, \pi(s))$ and $\overline{\mathcal{L}}_\pi y(s, \pi(s))$. In addition, the ratio of limits equals to the limit of the ratio. Thus, the long-term average reward in (\ref{eq:reward_theorem}) is well defined and exists.
\end{proof}

Next, in Theorem~\ref{theo:limitexists}, we prove that the underlying Markov chain is irreducible, and thus the long-term average date rate $\mathcal{R}(\pi)$ does not depend on the initial state $s_0$.

\begin{theorem}
	\label{theo:limitexists}
	For every $\pi$, the long-term average date rate $\mathcal{R}(\pi)$ is well defined and does not depend on the initial state, i.e., $\mathcal{R}_{\pi}(s)=\mathcal{R}_{\pi}, \forall s \in \mathcal{S}$.
\end{theorem}
\begin{proof}
The proof of Theorem~\ref{theo:limitexists} is provided in Appendix~\ref{appendix:recurrent}.
\end{proof}

Then, the long-term average data rate optimization problem can be formulated as follows:
\begin{eqnarray}
\label{eq:optimization}
\max_{\pi}	& &	\mathcal{R}_{\pi} = \frac{\overline{\mathcal{L}}_\pi r(s,\pi(s))}{\overline{\mathcal{L}}_\pi y(s,\pi(s))}	\\
\mbox{s.t.}			& & \sum_{s'\in \mathcal{S}}^{}\overline{\mathcal{L}}_\pi(s'|s) = 1, \forall s \in \mathcal{S} \nonumber.
\end{eqnarray}

Our aim in this work is finding the optimal beam association policy to maximize the long-term average data rate, i.e.,
\begin{equation}
\pi^* = \argmax_\pi \mathcal{R}_{\pi}.
\end{equation}

\section{Parallel Reinforcement Learning for Beam Association in High Mobility mmWave Vehicular Networks}
\label{sec:Q}

In this section, we develop the parallel QL algorithm that obtains the optimal beam association policy much faster than those of the existing reinforcement learning based algorithms (e.g., \cite{Khan2019Reinforcement}). For that, we first briefly present the details of the conventional QL algorithm. Related mmWave works in the literature usually adopt the QL and deep QL algorithms to solve their problems. However, with dynamic and complicated system, the QL algorithm usually takes a very long time to obtain the optimal strategy. In addition, the deep QL algorithm (e.g., \cite{Khan2019Reinforcement}) requires high performance computing resources and does not always ensure to converge to the optimal policy due to the overestimation of the optimizer. Note that our proposed parallel QL algorithm obtains the optimal beam association policy in an on-line manner, bringing various advantages, compared with the off-line learning approach. First, in mmWave vehicular communication systems, the conditions of the environment may be changed quickly. This is due to the fast and intermittent change in the quality of the mmWave links and the dynamics of blockages. Thus, a good beam/mmBS at the current time may quickly become worse later. For that, if we use the off-line training, the framework may not be able to obtain the ``current'' optimal policy in time, and thus leading to poor learning performance, especially in dynamic environments like mmWave networks. Under the proposed on-line training, channel conditions can be quickly learned to obtain the optimal beam association policy for vehicles. Second, off-line training requires frequent maintenance and adjustment to adapt to the new conditions of the environment (e.g., new buildings, new bus/transportation schedules). Differently, our on-line training can reduce human intervention and management costs as the algorithm can automatically learn the environment conditions and adjust its optimal policy.

\subsection{Q-Learning Approach}
This section presents the QL algorithm~\cite{Watkins1992QLearning}, which enable the eNodeB to obtain the optimal beam association strategy for vehicles without prior environment parameters, e.g., RSSI profiles and blockages. The key idea of the QL algorithm is updating the Q-value function for all state-action pairs stored in a Q-table. At a given system state, the QL algorithm performs an action and observes the immediate reward as well as the next state of the system. Based on these observations, the algorithm can update the Q-value for the current state-action pair based on the Q-value function~\cite{Watkins1992QLearning}. As such, the learning process is able to learn from the previous experiences, i.e., current state, action, next state, and immediate reward, to derive the optimal solution~\cite{Watkins1992QLearning}. In the following, we present the fundamentals of the Q-value function.

We first define the beam association policy as $\pi:\mathcal{S} \rightarrow \mathcal{A}$. In particular, $\pi$ is a mapping from a given state to its corresponding action. Our aim in this paper is finding the optimal beam association policy $\pi^*$ to optimize the system performance in terms of the average data rate, disconnection probability, and number of handovers. Then, we define $\mathcal{V}^\pi(s): \mathcal{S} \rightarrow \mathbb{R}$ as the expected value function of state $s \in \mathcal{S}$ given policy $\pi$. $\mathcal{V}^\pi(s)$ can be formulated as follows:
\begin{equation}
\begin{aligned}
\mathcal{V}^\pi(s) &= \mathbb{E}_\pi \Big [ \sum_{t=0}^{\infty} \gamma^t r_t(s_t, a_t)|s_0=s\Big ]\\
& =\mathbb{E}_\pi\Big [ r_t(s_t, a_t) + \gamma\mathcal{V}^\pi(s_{t+1})|s_0=s\Big ],
\end{aligned}
\end{equation}
where $0\leq \gamma < 1$ denotes the discount factor. In particular, $\gamma$ represents the effect of the future rewards. The higher the discount factor is, the more important future rewards are. At each state $s$, we aim to find the optimal action to derive the optimal beam association policy $\pi^*$, which is a map from a given state to the optimal action. To do that, the optimal value function for each state has to be obtained as in (\ref{eq:13}).
\begin{equation}
	\label{eq:13}
\mathcal{V}^*(s) = \max_{a} \Big \{ \mathbb{E}_\pi[r_t(s_t, a_t)+ \gamma\mathcal{V}^\pi(s_{t+1})] \Big\} ,\quad \forall s \in \mathcal{S}.
\end{equation}
We then denote the optimal Q-function state-action pair $(s,a), \forall s \in \mathcal{S}, \forall a \in \mathcal{A}$ as follows:
\begin{equation}
\mathcal{Q}^*(s,a) \triangleq r_t(s_t, a_t) + \gamma\mathbb{E}_\pi[\mathcal{V}^\pi(s_{t+1})].
\end{equation}
Hence, the optimal value function is written as follows:
\begin{equation}
\label{eq:valuefunction}
\mathcal{V}^*(s) = \max_{a} \{ \mathcal{Q}^*(s,a)\}.
\end{equation}
To solve~(\ref{eq:valuefunction}), we can update the Q-function to determine the optimal Q-values of all state-action pairs by using (\ref{Eq:qfunction})~\cite{Watkins1992QLearning}:
\begin{equation}
\label{Eq:qfunction}
\begin{aligned}
\mathcal{Q}_{t+1}(s_t,a_t) = \mathcal{Q}_t&(s_t,a_t) + \tau_t \Big [ r_t(s_t, a_t)\\
& + \gamma\max_{a_{t+1}} \mathcal{Q}_t(s_{t+1}, a_{t+1})- \mathcal{Q}_t(s_t,a_t)\Big ],
\end{aligned}
\end{equation}
where $\tau_t$ denotes the learning rate, determining the impact of new experiences to the current Q-value~\cite{Watkins1992QLearning}. By updating the Q-value functions of all state-action pairs by using (\ref{Eq:qfunction}), the algorithm can derive the optimal beam association policy.

\subsection{Parallel Q-Learning Approach}
\label{subsec:parallelQ}
\begin{figure*}[!]
	\centering
	\includegraphics[scale=0.33]{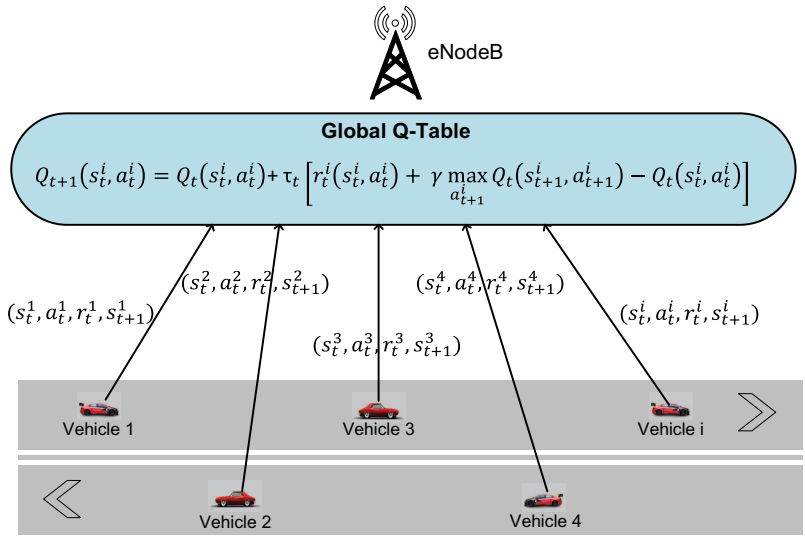}
	\caption{Parallel QL based model enables simultaneous learning from multiple vehicles.}
	\label{fig:parallellearning}
\end{figure*}

Note that the conventional QL algorithm can converge to the optimal beam association policy quickly when the system is simple. However, with the dynamics and uncertainties of the system considered in this work, the QL algorithm may take a very long time to obtain the optimal strategy. This is due to the fact that the QL algorithm require a huge number of training episodes to collect enough data for learning. To speed up the QL algorithm, the deep reinforcement learning algorithm is usually adopted in the literature. Nevertheless, this algorithm requires high performance computing resources and does not ensure to converge to the optimal policy due to the overestimation of the optimizer~\cite{JSAC2},~\cite{JSAC}. The convergence rate of the deep reinforcement learning algorithm can be improved by implementing several techniques. For example, by using the deep dueling neural network architecture~\cite{Wang2015Dueling}, one can significantly improve the convergence rate and stability of the QL algorithm. However, in our work, we leverage the fact that there are simultaneously multiple vehicles running on the road to design the parallel QL algorithm. The proposed algorithm does not only require lower complexity but also converge faster than the latest deep learning-based approaches. In particular, vehicles running on the road act as active learners which can help the system simultaneously collect data and significantly speed up the learning process as shown in Fig.~\ref{fig:parallellearning}.

To that end, the parallel QL algorithm employs multiple learning processes. Each learning process is assigned for a vehicle running on the road (in the coverage of the eNodeB). Specifically, each learning process $i$ updates the Q-value function at the global Q-table as follows:
\begin{equation}
\label{eq:updateQ}
\begin{aligned}
\mathcal{Q}_{t+1}(s_t^i,a_t^i) = \mathcal{Q}_t&(s_t^i,a_t^i) + \tau_t \Big [ r_t^i(s_t^i, a_t^i) \\
&+ \gamma\max_{a_{t+1}^i} \mathcal{Q}_t(s_{t+1}^i, a_{t+1}^i)- \mathcal{Q}_t(s_t^i,a_t^i)\Big ],
\end{aligned}
\end{equation}
where $0\leq \gamma < 1$ is the discount factor that presents the effect of future rewards~\cite{Watkins1992QLearning}. In particular, when $\gamma$ is low, e.g., close to 0, the learning process prefers the current reward. Differently, when $\gamma$ is high, e.g., close to 1, the long-term reward will be considered. In this work, we set $\gamma$ the same for all vehicles, i.e., learners. $r_t^i(s_t^i, a_t^i)$ is the immediate reward when vehicle $i$ performs action $a_t^i$ at state $s_t^i$ (computed using equation \ref{5} above). $\tau_t$ is the learning rate at decision epoch $t$~\cite{Watkins1992QLearning}. Note that the learning rate can be fixed at a constant value or it can be adjusted when running the algorithm. In this paper, the learning rate is fixed during the training process and is the same for all learning processes. At each decision epoch and given a current state, i.e., RSSI level and current connected beam, the current vehicle chooses to connect to a beam (following the current beam association policy sent from the eNodeB) and observes the data rate of the connected beam as well as the next state. Then, these observations are sent to the eNodeB for learning by the corresponding learning process to update the global Q-table (equation \ref{eq:updateQ}). Algorithm~\ref{algorithm1} describes the proposed parallel QL algorithm.
\begin{algorithm}
	\caption{Parallel QL Algorithm for Vehicle $i$}
	\label{algorithm1}
	\begin{algorithmic}[1]
		\For{\textit{t=1 to T}}
		\State Vehicle $i$ observes the current state $s_t^i \in \mathcal{S}$ and execute action $a_t^i \in \mathcal{A}$ based on the $\epsilon$-greedy policy.
		\State Vehicle $i$ observes the immediate reward $r_t^i$ and new state $s_{t+1}^i \in \mathcal{S}$.
		\State Vehicle $i$ sends transition $(s_t^i, a_t^i, r_t^i, s_{t+1}^i)$ to the eNodeB for learning by updating the table entry of $\mathcal{Q}(s_t^i, a_t^i)$ as as follows:
		\begin{equation}
		\begin{aligned}
		\mathcal{Q}_{t+1}(s_t^i,a_t^i) = \mathcal{Q}_t&(s_t^i,a_t^i) + \tau_t\Big [ r_t^i(s_t^i, a_t^i) \\
		&+ \gamma\max_{a_{t+1}^i} \mathcal{Q}_t(s_{t+1}^i, a_{t+1}^i)- \mathcal{Q}_t(s_t^i,a_t^i)\Big ]
		\end{aligned}
		\end{equation}
		\State Replace $s_t^i \leftarrow s_{t+1}^i$.
		\EndFor
	\end{algorithmic}
\end{algorithm}

In particular, vehicle $i$ first observes the current state $s_t^i \in \mathcal{S}$ and performs action $a_t^i$ based on the $\epsilon$-greedy policy~\cite{Sutton1998Reinforcement}. Then, the eNodeB selects an action that maximizes the Q-value function with probability $1-\epsilon$ and a random action with probability $\epsilon$. Then the eNodeB sends this action to vehicle $i$ to perform. In this work, we gradually reduce the value of $\epsilon$. In other words, the algorithm first chooses random actions and gradually change to the deterministic strategy, i.e., choose an action with the highest Q-value at a given state. To that end, $\epsilon$ is set at a high value (e.g., 1) when the algorithm starts running. Then, at later iterations, the value of $\epsilon$ is slowly reduced to a small value (e.g., 0.01). After performing action $a_t^i$, vehicle $i$ observes immediate data rate $r_t^i(s_t^i, a_t^i)$ and next state $s_{t+1}^i$. These observations are then sent to the eNodeB for learning. Note that the learning process of each vehicle is independent from others, and all the learning processes share the same global Q-table. By doing that, the Q-table is updated with more experiences from multiple vehicles running on the road. As such, the convergence rate and convergence time of the parallel QL algorithm will be better than that of the conventional QL algorithm as demonstrated in Section~\ref{sec:evaluation}.

The convergence of our proposed parallel QL algorithm can be always guaranteed. In particular, as the learning processes update the global Q-table simultaneously, the selection of an action of a learning process may be based on a new Q-value that is updated by other learning processes. Specifically, at state $s$, a learning process chooses to perform action $a$ that maximizes the Q-value function. However, before that instance, another learning process may reach state $s$, then take action $a'$, and update the Q-value of action $a'$ to be the highest value. In this case, action $a$ may still be selected by the $\epsilon$-greedy policy mentioned above. In other words, an iteration for exploitation becomes an iteration for exploration. Moreover, multiple learning processes (corresponding to multiple vehicles running on the road), which simultaneously update different non-related states, can be considered as independent processes~\cite{Camelo2016Scalable}. As a result, the convergence of the proposed parallel Q-learning algorithm is guaranteed and sped up. Given the above, the convergence of the proposed parallel QL algorithm to the optimal policy is formally stated in Theorem~\ref{theo:convergeQ}.
\begin{theorem}
	\label{theo:convergeQ}
	Given the learning processes are serializable and under the conditions of $\tau_t$ in~(\ref{Eq:rules}), the parallel QL algorithm is ensured to converge to the optimal policy.
	\begin{equation}
		\label{Eq:rules}
		\tau_t \in [0,1), \sum_{t=1}^{\infty}\tau_t = \infty, \mbox{ and } \sum_{t=1}^{\infty} ( \tau_t  )^{2} < \infty.
	\end{equation}
\end{theorem}
\begin{proof}
Theorem~\ref{theo:convergeQ} is proved in Appendix~\ref{appendix:convergeQ}.
\end{proof}

\subsection{Impact of High Mobility on Convergence Time}
\label{subsec:convergencetime}
As mentioned, the proposed parallel QL is particularly useful in vehicular networks where the system can simultaneously learn from experiences of multiple vehicles. Moreover, the high mobility of vehicles is also exploited in collecting new experiences in our proposed framework. In particular, the duration to the next decision epoch can be calculated as follows:
\begin{equation}
	\label{eq:timeepoch1}
	\xi_t = \frac{z}{v_{t}},
\end{equation}
where $v_t$ is the speed of the vehicle between the $t$-th and $(t+1)$-th decision epochs. After a period of $\xi_t$, the vehicle can collect a new sample of experience at decision epoch $(t+1)$-th, i.e., $(s^i_{t+1},a^i_{t+1}, r^i_{t+1},s^i_{t+2})$. From (\ref{eq:timeepoch1}), it is clear that with high speeds, the time to move to the next decision epoch of the vehicle will be short. As such, with a given time period, with higher speeds, the vehicle can collect more experiences. Thus, the parallel QL algorithm can converge to the optimal policy faster. In the simulations below, we can observe that the convergence time reduces from $6,000$s to $1,000$s when the speed increases from 2 m/s to 9 m/s.

\subsection{Complexity and Overhead of Parallel Q-Learning}

The proposed parallel QL algorithm is efficient with low computational complexity and memory complexity. As mentioned, the state space of our system includes only the current RSSI level, the current connected beam, and the vehicle's velocity. For that, in common mmWave vehicular networks setting with a few mmBSs, the number of states is small, and thus the size of the global Q-table is also small. Hence, the algorithm can obtain the optimal beam association strategy quickly as the lookup and update table processes are very fast.

Regarding the computational complexity, our proposed algorithm only performs basic calculations without any complex functions as in the other reinforcement learning algorithms, e.g., deep double QL, deep QL, and deep dueling~\cite{JSAC2},~\cite{JSAC}. These algorithms implement deep neural networks to approximate the Q-value function to obtain the optimal policy with complicated mathematical operations, e.g., multiply matrices and gradient descent. As a result, they require longer time to process and higher computing resources compared to our proposed parallel QL algorithm. In addition, in this work, we deploy only one Q-table at the eNodeB to store the Q-values for all state-action pairs instead of implementing a separated Q-tables on each vehicle with limited resources. As a result, the computing complexity is moved to the eNodeB which has sufficient resources to obtain the optimal policy in a short time.

Finally, our proposed parallel QL algorithm incurs minimal communication overhead. In particular, the RSSI level can be inferred by the mmBS through the signals received from the vehicle. Moreover, the mmBS can always know which beam the vehicle connected to. These information are sent to the eNodeB through the backhaul link with high bandwidth. Furthermore, in intelligent transport systems \cite{Standard}, the location information of vehicles is frequently reported to the RSU, i.e., mmBS. Thus, to update event $e_s$ of the SMDP, the eNodeB can collect the information of each vehicle through the mmBSs. Therefore, our proposed solution does not add additional overheads to the current ITS standards.
\section{Performance Evaluation}
\label{sec:evaluation}

\subsection{Parameter Setting}
We consider a road with a length of $1000$ meters in the coverage of an eNodeB. On the considered road, $10$ mmBSs are deployed. Each mmBS is equipped with $10$ orthogonal beams. Each beam is assumed to cover an area (on the considered road) ranging randomly from $20$ meters to $50$ meters. The blocking probability (including both temporary and permanent blockages) of each beam is generated randomly from $0$ to $1$. We define $10$ RSSI levels for each beam corresponding to $10$ data rates ranging from $0$ to $R_{\max}=9$ Gbps, i.e., $r_{n,k} \in \{0, 1, \ldots, 9\}$. Unless otherwise stated, at each decision epoch, a vehicle enters the road with probability $\lambda = 0.5$. The handover time is set at $0.5$ seconds~\cite{Zang2019Managing},~\cite{Park2018Handover}. $z$ is set at $5$ meters. The average vehicle speed is set at $7$ m/s (about 25 km/h), which is a typical vehicle urban speed~\cite{Khan2019Reinforcement}. During a decision epoch, the vehicle's speed remains unchanged. The average speed is then varied from $1$ m/s to $9$ m/s (about 4 km/h to 32 km/h) in several scenarios to demonstrate the effectiveness of the proposed solution under high mobility. It is important to note that the proposed parallel QL algorithm can learn without requiring these parameters in advance. Instead, the algorithm will learn them by interacting with the environment. For the proposed parallel QL and QL algorithms, the learning rate and discount factor are set at $0.1$ and $0.9$, respectively. Moreover, for the $\epsilon$-greedy method, the initial value of $\epsilon$ is set at $1$ and gradually reduced to $0.01$.

We compare our proposed algorithm with three other methods: (i) \textit{MaxRate}, (ii) \textit{Blockage Aware}, and (iii) \textit{Upper Bound}.

\begin{itemize}
	\item \textbf{MaxRate:} This scheme first explores all available beams at the current location. Then, the beam with the highest RSSI level will be selected to connect. Once the \textit{MaxRate} scheme selects the best beam, it will keep connecting to this beam until the end of the current decision epoch. This scheme is used to show the performance of non-adaptive and greedy solutions. As in mmWave systems the temporal degradation of the channel quality frequently occurs, the best beam at the current time may become worse later. Thus, this scheme results in poor system performance in terms of data rate, the number of handovers, and disconnection probability.
	
	\item \textbf{Blockage Aware:} This scheme is assumed to know the prior knowledge about the blocking probability of all available beams in the current location. Then, the beam with the lowest blocking probability will be chosen for the vehicle to connect. This scheme is used to show the effects of the high mobility. Specifically, when the vehicle speed is low, this scheme will achieve a good system performance in terms of data rate, the number of handovers, and disconnection probability. However, its performance becomes worse when the vehicle speed is high as the high mobility is not considered.

	\item \textbf{Upper Bound:} Similar to the \textit{Blockage Aware} scheme, this scheme is assumed to know the prior knowledge about the blocking probability of all available beams in the current location. However, the handover decision is only made if the data rate achieved is higher than that of staying in the current beam. Note that in practice, this data rate is not available in advance and depended on the channel quality. This scheme is adopted to show the optimistic upper bound of the system performance.
\end{itemize}
The evaluation metrics are the average data rate, the disconnection probability, and the number of handovers. The average data rate is defined as the data received (in bits) by a vehicle running on the road in a second. The number of handovers is defined as the total number of handovers that a vehicle needs to do when running on the road. The disconnection probability is defined as the average probability that the vehicle cannot communicate with mmBSs.

\subsection{Simulation Results}
\label{sec:evaluationB}
\paragraph{Performance Evaluation}

We first vary the average speed of vehicles running on the road and evaluate the system performance obtained by the parallel QL algorithm in terms of the average data rate, the disconnection probability, and the number of handovers as shown in Fig.~\ref{fig:varySpeed}. Obviously, the average data rate of the vehicle decreases when the vehicle speed increases as shown in Fig.~\ref{fig:varySpeed}(a). The reason is due to the effects of mobility on the handover. In particular, in the case that the vehicle chooses to handover to a new beam and its speed is high, the vehicle may move to another beam before the handover is finished. Moreover, at high speeds, the estimated RSSI level can be less accurate. This results in lower data rates. It is worth noting that by learning the environment parameters, our proposed solution achieves higher data rates than that of the \textit{MaxRate} scheme and the achieved rate is close to the \textit{Upper Bound} scheme. Moreover, when the vehicle speed is high (i.e., $\geq 7$ m/s), our proposed solution achieves a higher data rate than that of the \textit{Blockage Aware} scheme. The reason is that the \textit{Blockage Aware} scheme selects the beam with the lowest blocking probability at a given location without considering the vehicle speed. However, at high speeds, the beam with the lowest blocking probability may not be the best choice as the vehicle may move out of the beam before completing the handover process, and thus resulting in low data rates. As can be seen in Fig.~\ref{fig:varySpeed}(b), when the speed increases from 1 m/s to 6 m/s, our proposed solution chooses to reduce the number of handovers to avoid the negative effect of mobility. In contrast, other solutions with fixed policies cannot learn this information, and thus they do not reduce the number of handovers, except the \textit{Upper Bound} scheme which is assumed to know the actual data rate before connecting to a beam. It is worth noting that, when the speed is higher than $7$ m/s (about 25 km/h), the number of handovers slightly increases for all schemes. This is because at high speeds, the vehicle moves out the coverage of a beam before finishing the handover to connect to this beam, thereby it needs to do the handover again. It is important to note that, at several locations, the only option is to handover to a new beam. Thus, the number of handovers slightly increases for our proposed solution in this case. Finally, as shown in Fig.~\ref{fig:varySpeed}(c), the average disconnection probabilities of all schemes increase when the vehicle's speed increases. The reason is due to the effects of the high mobility on the handover process. Note that by using the learning algorithm, our proposed solution achieves lower disconnection probability than that of the \textit{MaxRate} scheme and the probability is very close to that of the \textit{Upper Bound} scheme, especially at high speeds.

\begin{figure*}[!]
	\centering
	\begin{subfigure}[b]{0.33\textwidth}
		\centering
		\includegraphics[scale=0.3]{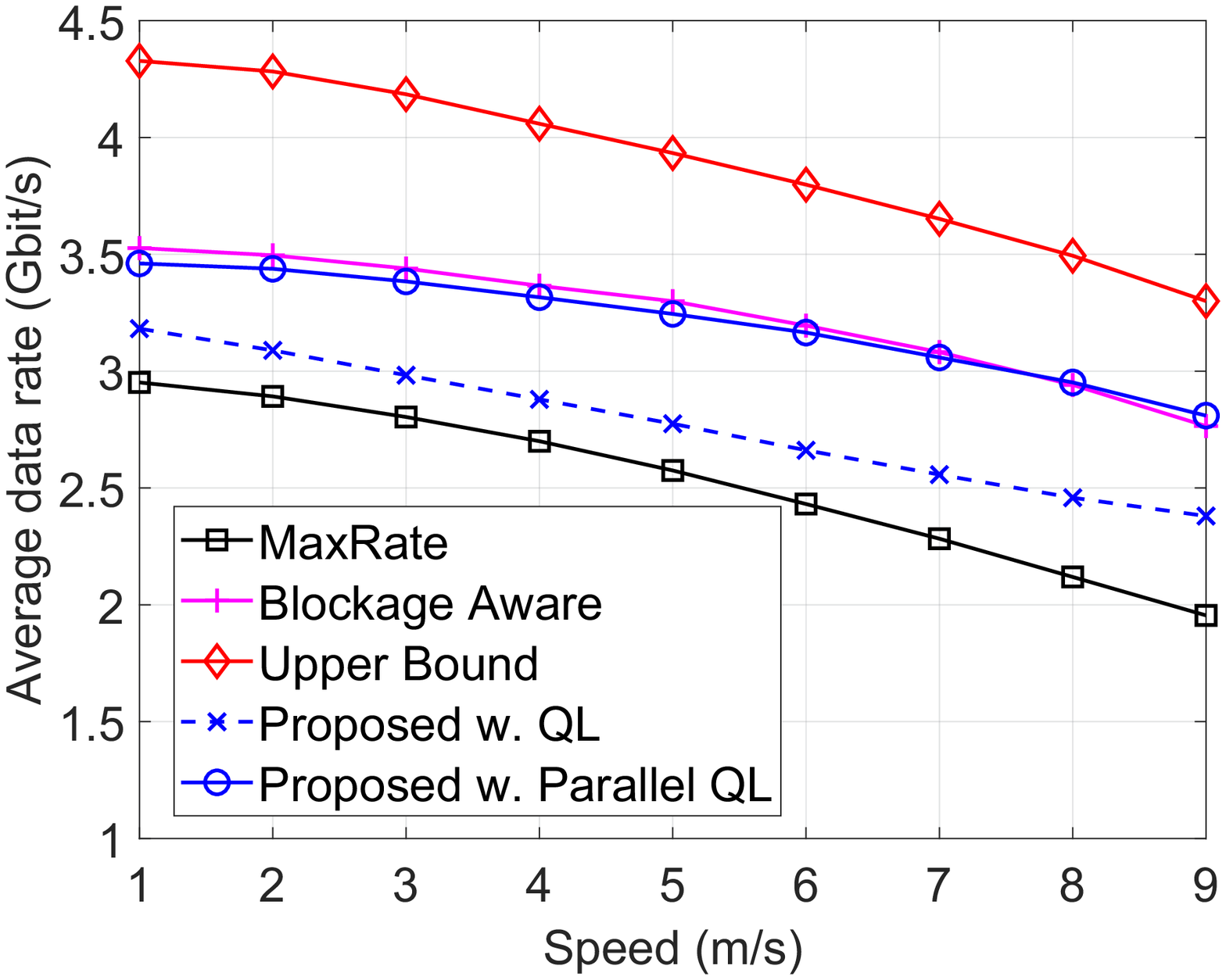}
		\caption{}
	\end{subfigure}%
	~
	\begin{subfigure}[b]{0.33\textwidth}
		\centering
		\includegraphics[scale=0.3]{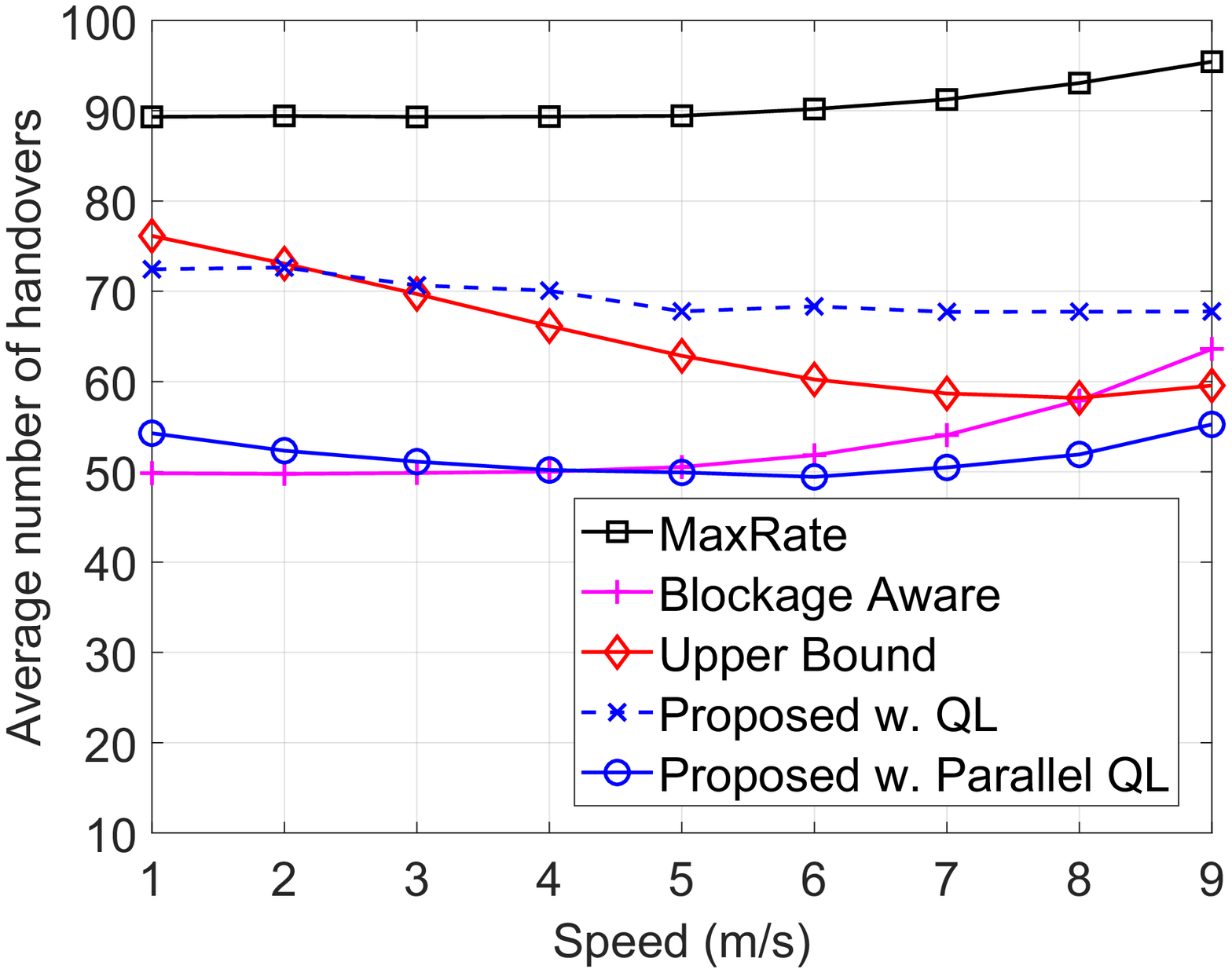}
		\caption{}
	\end{subfigure}%
	~
	\begin{subfigure}[b]{0.33\textwidth}
		\centering
		\includegraphics[scale=0.3]{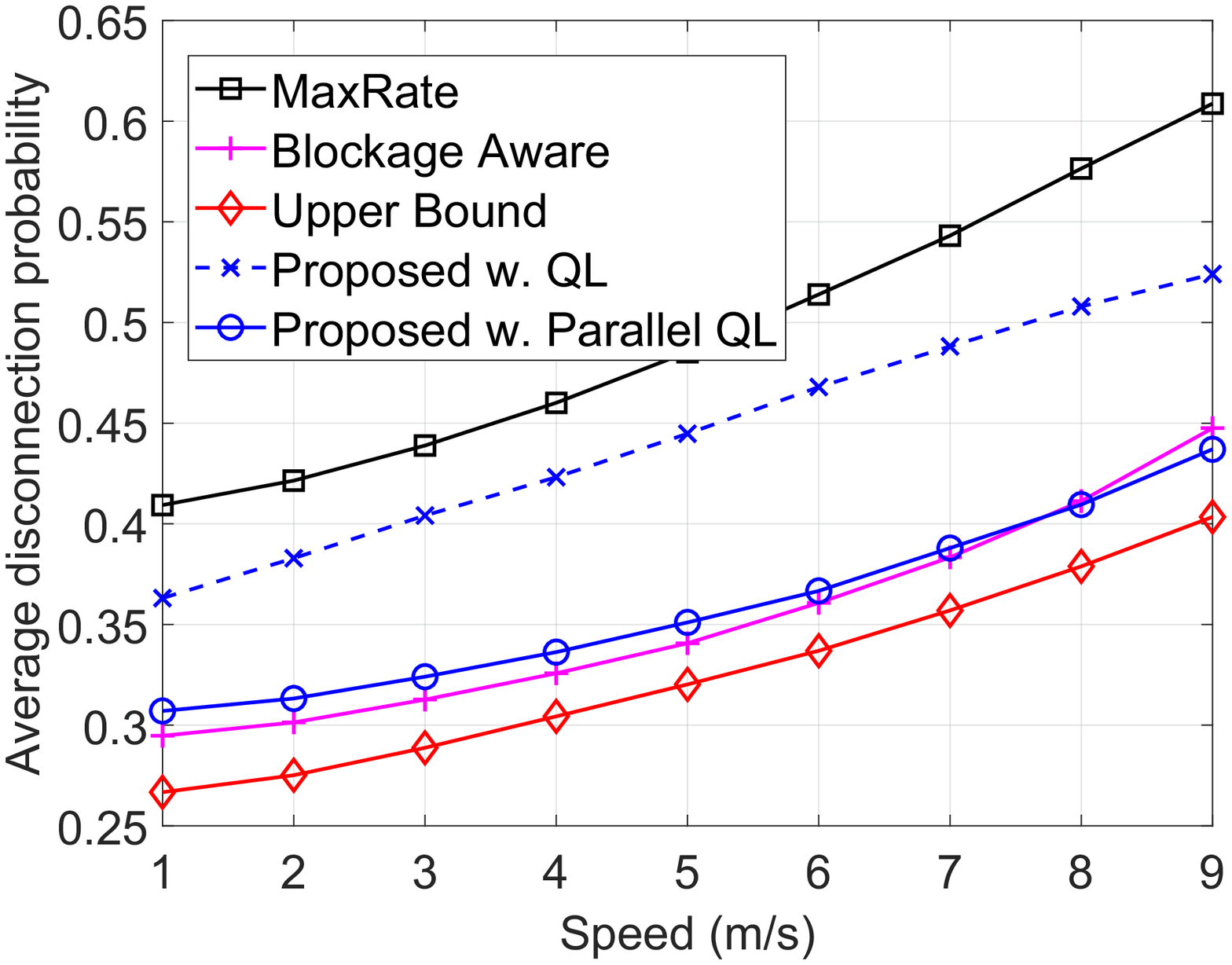}
		\caption{}
	\end{subfigure}%
	
	\caption{(a) Average data rate (Gbit/s), (b) average number of handovers, and (c) average disconnection probability vs. average speed of vehicles.} 
	\label{fig:varySpeed}
\end{figure*}

\begin{figure*}[!]
	\centering
	\begin{subfigure}[b]{0.33\textwidth}
		\centering
		\includegraphics[scale=0.3]{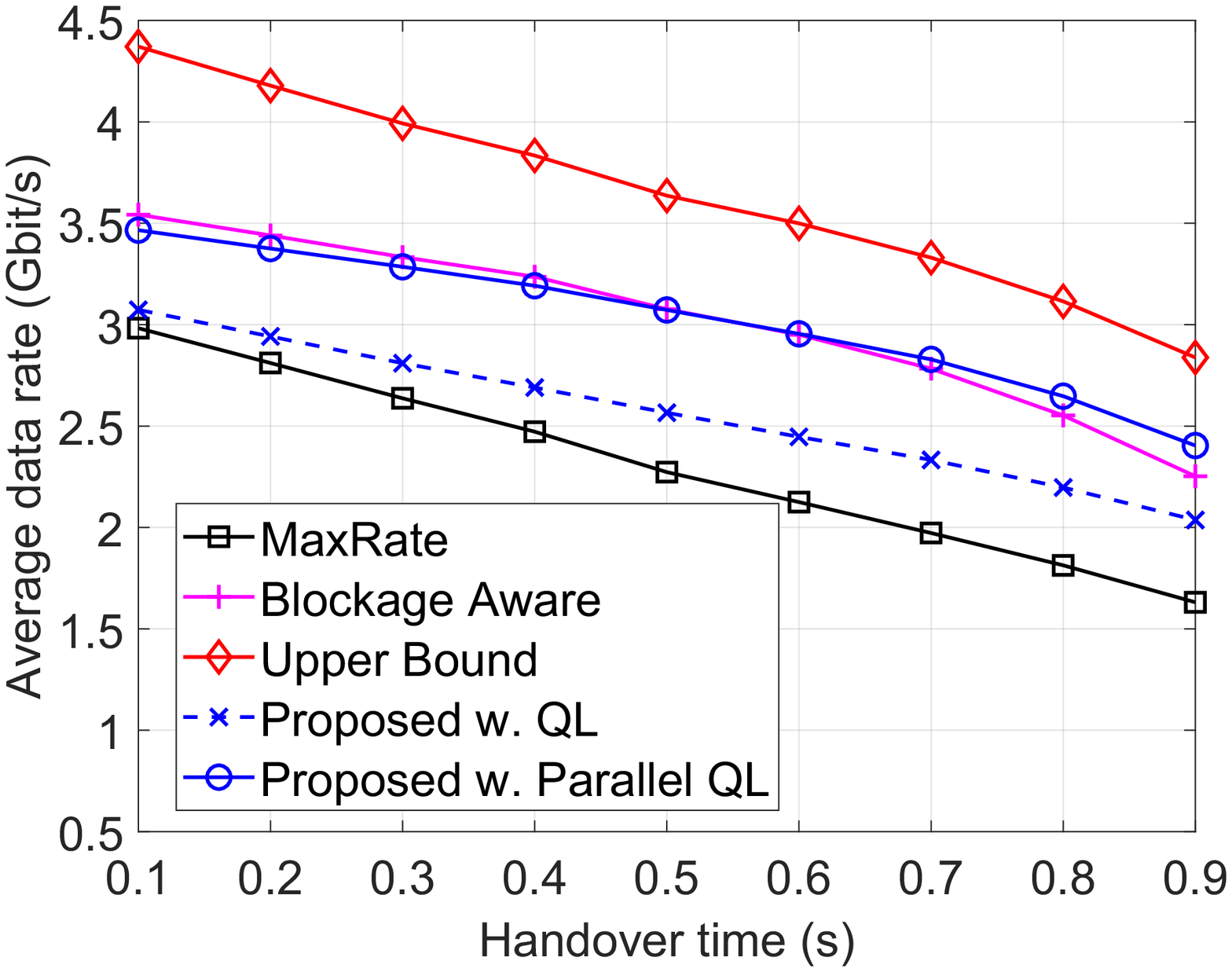}
		\caption{}
	\end{subfigure}%
	~
	\begin{subfigure}[b]{0.33\textwidth}
		\centering
		\includegraphics[scale=0.3]{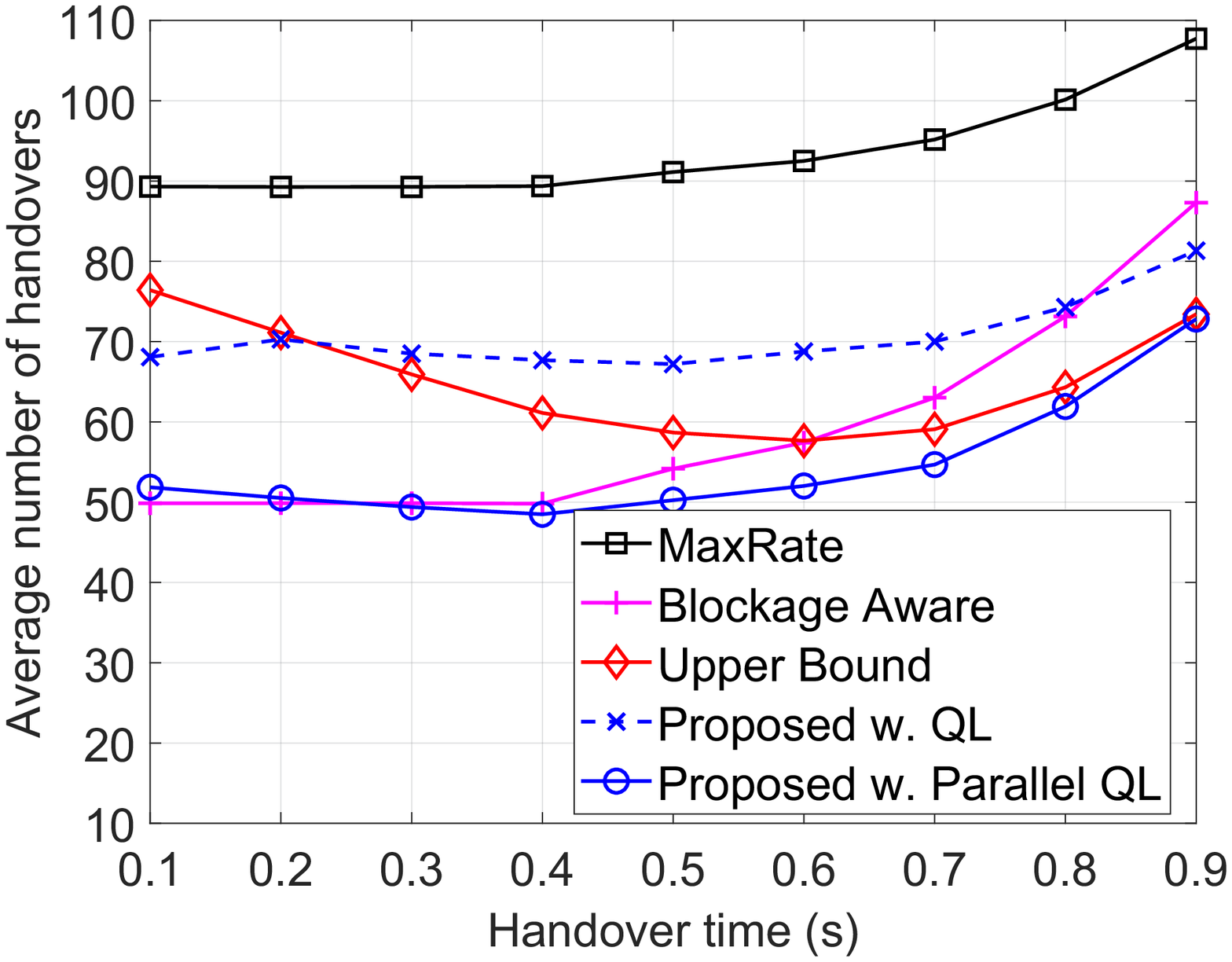}
		\caption{}
	\end{subfigure}%
	~
	\begin{subfigure}[b]{0.33\textwidth}
		\centering
		\includegraphics[scale=0.3]{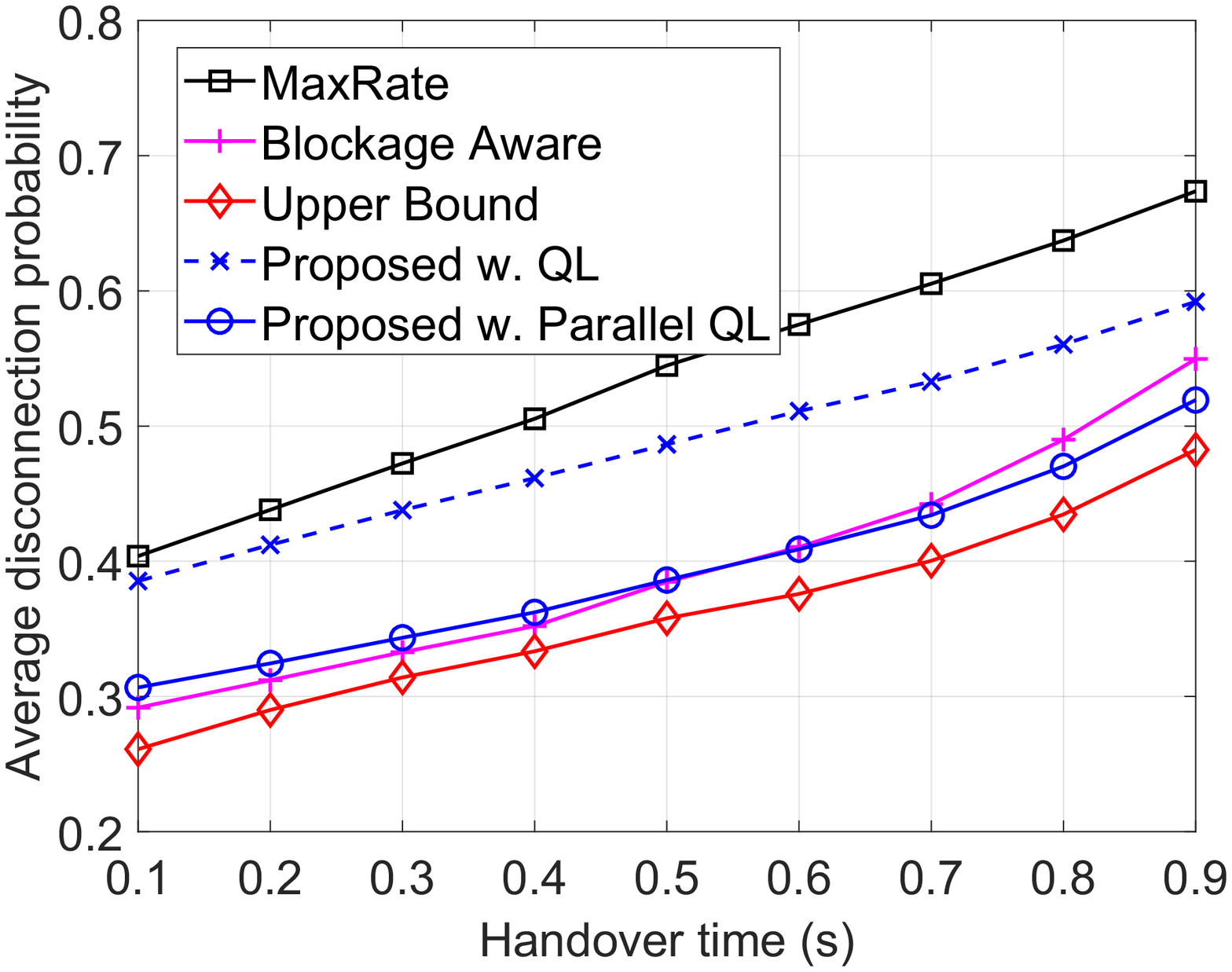}
		\caption{}
	\end{subfigure}%
	
	\caption{(a) Average data rate (Gbit/s), (b) average number of handovers, and (c) average disconnection probability vs. $h$.} 
	\label{fig:varyHandovertime}
\end{figure*}

Next, we fix the average speed of vehicles at $7$ m/s (about 25 km/h), which is a typical vehicle urban speed~\cite{Khan2019Reinforcement} and vary the time for the handover to show the average data rate, number of handovers, and disconnection probability obtained by the proposed parallel QL algorithm as shown in Fig.~\ref{fig:varyHandovertime}. Clearly, when the time for the handover increases, the disconnection probability increases, and thus the average data rates of all solution decrease as shown in Fig.~\ref{fig:varyHandovertime}(a) and Fig.~\ref{fig:varyHandovertime}(c). Again, the proposed solution possesses better performance in terms of data rate and disconnection probability than those of the \textit{MaxRate} scheme and close to that of the \textit{Upper Bound} solution. The reason is that our proposed solution can learn and minimize the number of handovers when the time for the handover increases. In particular, as shown in Fig.~\ref{fig:varyHandovertime}(b), when the time for handover increases from $0.1$ seconds to $0.4$ seconds, the number of handovers of the proposed solution decreases as our learning algorithm adapts its policy to minimize the number of the handovers and thus maximize the data rate for the vehicle (when the environment parameters are not available in advance). However, when the time for the handovers is too long (i.e., $\geq 0.5$), the number of handovers increases. Similar to the previous scenario, when the handover is too long, the vehicle may move out of the target beam, and thus it needs to do the handover again. It is worth noting that at high speeds, our proposed solution performs betters than the \textit{Blockage Aware} scheme with prior knowledge about the beam profiles. The reason is that the \textit{Blockage Aware} scheme selects a beam based on its blocking probability without considering the handover process.

\begin{figure*}[!]
	\centering
	\begin{subfigure}[b]{0.33\textwidth}
		\centering
		\includegraphics[scale=0.3]{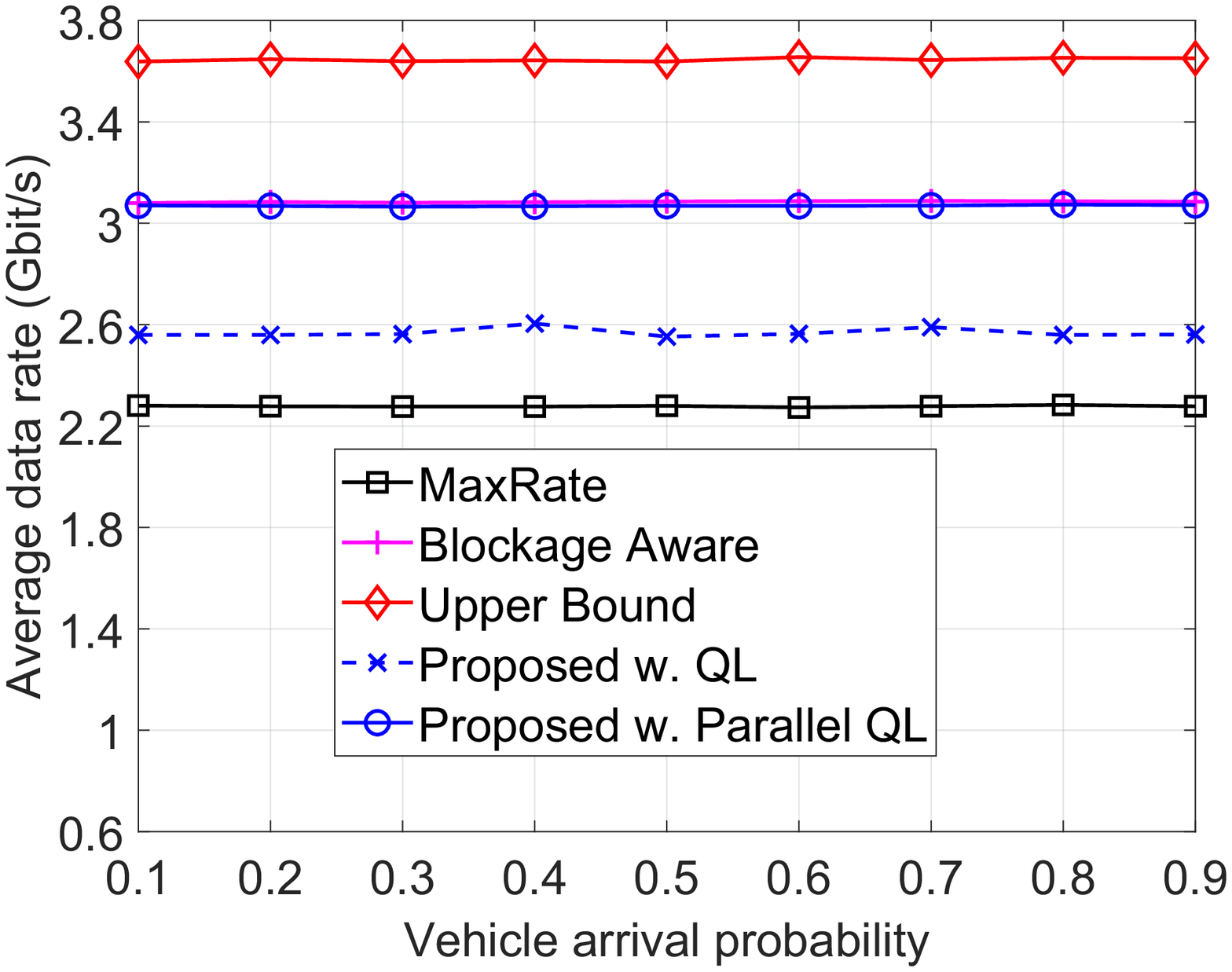}
		\caption{}
	\end{subfigure}%
	~
	\begin{subfigure}[b]{0.33\textwidth}
		\centering
		\includegraphics[scale=0.3]{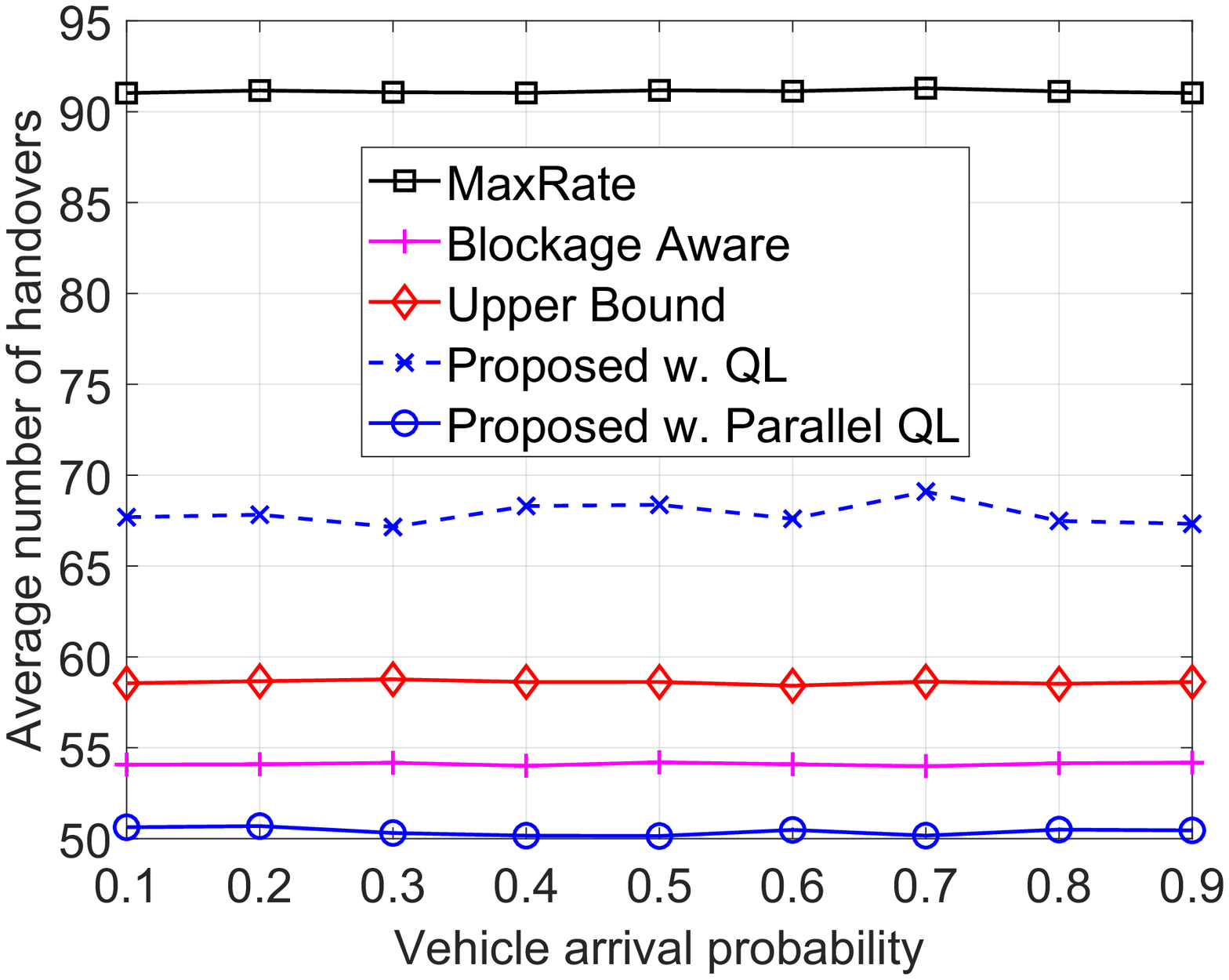}
		\caption{}
	\end{subfigure}%
	~
	\begin{subfigure}[b]{0.33\textwidth}
		\centering
		\includegraphics[scale=0.3]{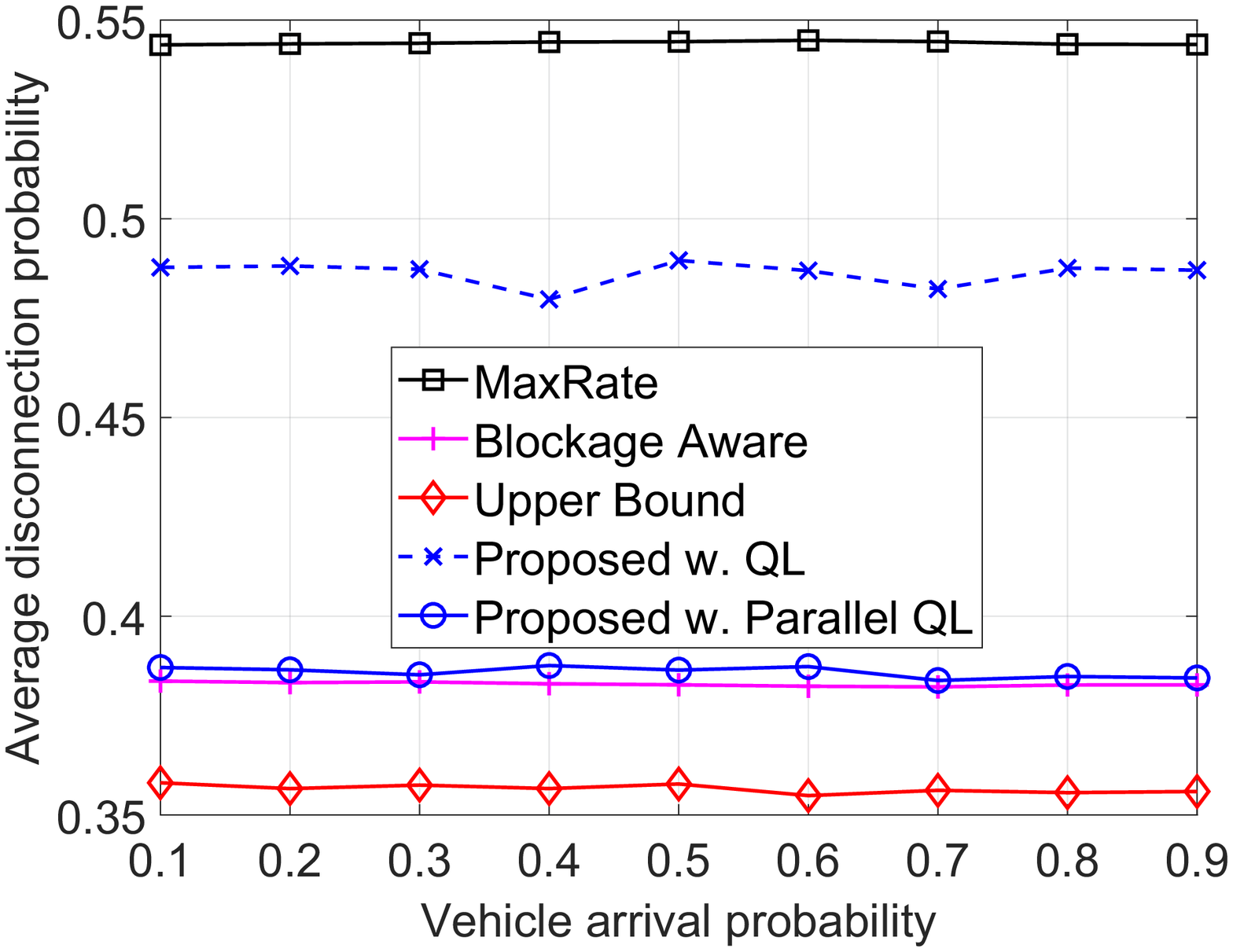}
		\caption{}
	\end{subfigure}%
	
	\caption{(a) Average data rate (Gbit/s), (b) average number of handovers, and (c) average disconnection probability vs. $\lambda$.} 
	\label{fig:varyCarArrival}
\end{figure*}

Finally, in Fig.~\ref{fig:varyCarArrival}, we vary the probability that a vehicle enters the road at each decision epoch to evaluate the performance of the proposed solution. Similar to other scenarios, by learning the environment parameters, i.e., blocking probability, vehicle's speed, and handover time, our proposed solution achieves better performance than that of the \textit{MaxRate} and \textit{Blockage Aware} schemes in terms of data rate, number of handovers, and disconnection probability.

\paragraph{Convergence}
\begin{figure}[h]
	\centering
	\includegraphics[scale=0.33]{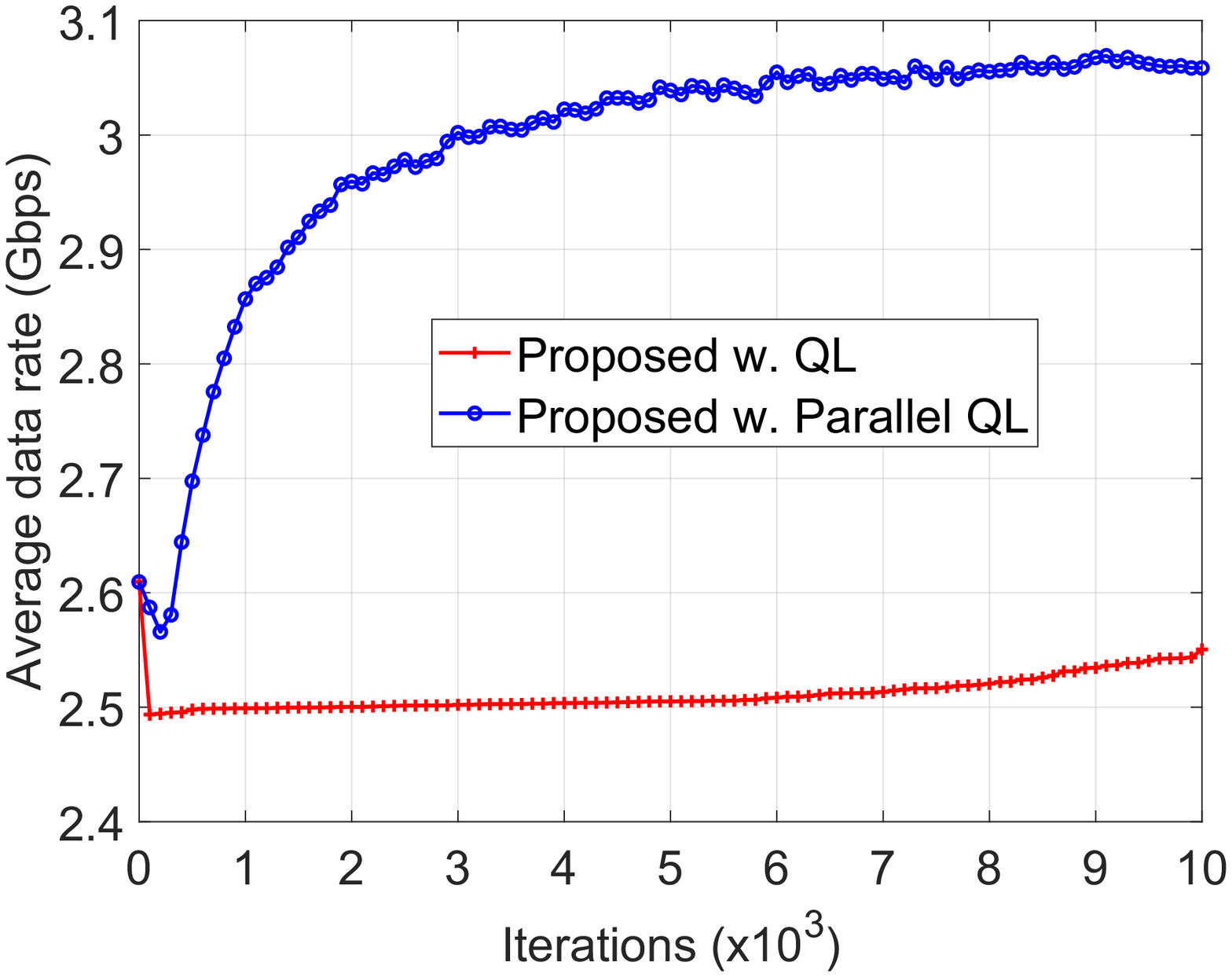}
	\caption{Convergence rates of QL and parallel QL algorithms.}
	\label{fig:convergence10}
\end{figure}

In Fig.~\ref{fig:convergence10}, we evaluate the convergence rates of the proposed parallel QL and the QL algorithms. Obviously, the parallel QL algorithm can obtain the optimal beam association policy within $6,000$ iterations while the QL algorithm cannot converge to the optimal beam association policy after $10,000$ iterations. This result confirms the analysis in Section~\ref{subsec:parallelQ}. Specifically, by learning from multiple vehicles on the road at the same time, the proposed parallel algorithm has more experiences to learn and quickly converge to the optimal policy.

\begin{figure}[h]
	\centering
	\includegraphics[scale=0.33]{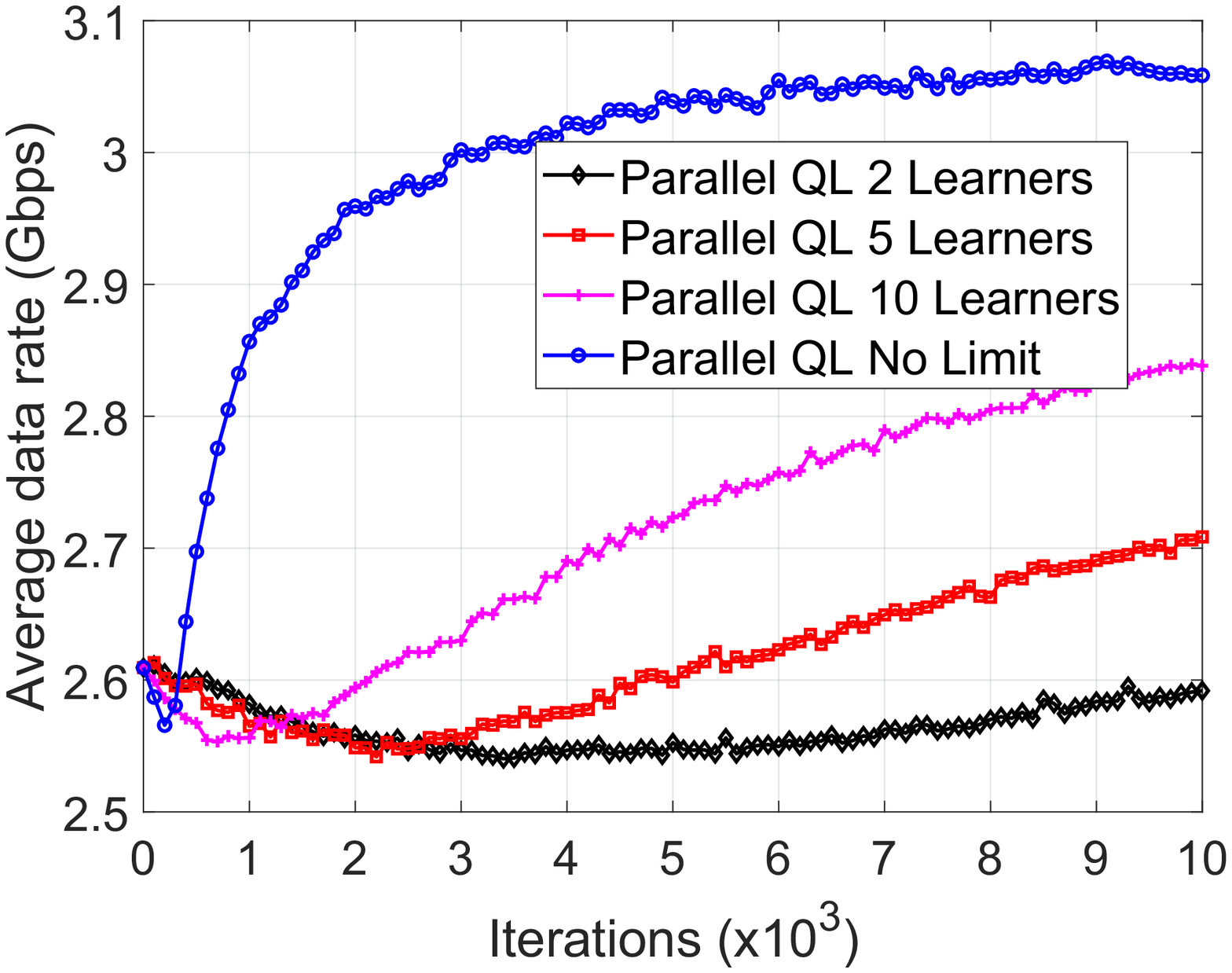}
	\caption{Convergence rates of the parallel QL algorithm with different numbers of learners.}
	\label{fig:compare}
\end{figure}

Next, in Fig.~\ref{fig:compare}, we compare the convergence rates of the parallel QL with different numbers of learners. Clearly, the higher number of learners results in better performance. In particular, when the parallel QL runs with only 2 learners, the performance is the worst. When the number of learners increases to 10, the performance of the algorithm is improved. Finally, in the case that we do not limit the number of learners (i.e., learning from all the vehicles running on the road), the algorithm achieves the best performance and quickly convergences to the optimal beam association policy. This implies that by leveraging the fact that there are often numerous vehicles on the road, our proposed parallel QL algorithm can significantly improve the system performance compared to conventional methods.

\begin{figure}[h]
	\centering
	\begin{subfigure}[b]{0.33\textwidth}
		\centering
		\includegraphics[scale=0.33]{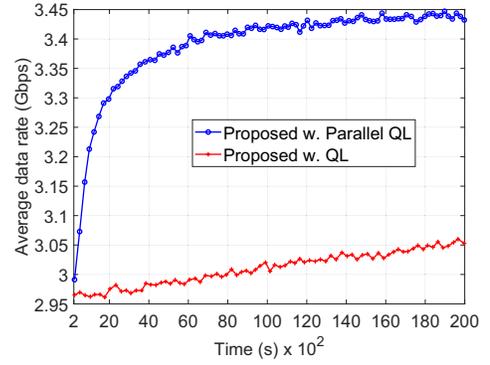}
		\caption{}
	\end{subfigure}%
	
	\begin{subfigure}[b]{0.33\textwidth}
		\centering
		\includegraphics[scale=0.33]{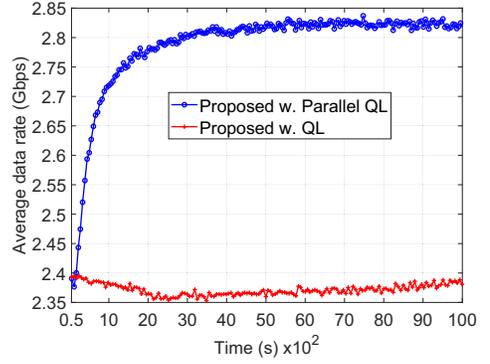}
		\caption{}
	\end{subfigure}	
	\caption{Convergence rates of the QL and parallel QL algorithms when the average speed of vehicles is (a) 2 m/s and (b) 9 m/s.} 
	\label{fig:ConvergenceSpeed}
\end{figure}

Next, we investigate the convergence time of the proposed parallel QL and the QL algorithms when the average speed of vehicles is low (2 m/s) and high (9 m/s) as shown in Fig.~\ref{fig:ConvergenceSpeed}(a) and Fig.~\ref{fig:ConvergenceSpeed}(b), respectively. As discussed in Section~\ref{subsec:convergencetime}, increasing the speed of vehicles leads to better convergence time of the algorithm as vehicles can collect more experiences for the learning process. As shown in Fig.~\ref{fig:ConvergenceSpeed}(a), when the speed of vehicles is 2 m/s, the parallel QL algorithm requires at least $14,000$ seconds to obtain the optimal beam association/handover policy. In contrast, when the speed of vehicles is 9 m/s, the algorithm can obtain the optimal association policy within $3,000$ seconds as shown in Fig.~\ref{fig:ConvergenceSpeed}(b). Meanwhile, the QL algorithm still cannot converge to the optimal policy after $20,000$ seconds and $10,000$ seconds when the vehicle speed is 2 m/s and 9 m/s, respectively. Note that the average data rate (after obtaining the optimal beam association strategy) achieved by the proposed solution in the case the speed of vehicles is 9 m/s lower than that of the case when the speed of vehicles is 2 m/s. As mentioned, this is stemmed from the effects of mobility on the handover process as discussed in the above section.

\begin{figure}[h]
	\centering
	\includegraphics[scale=0.33]{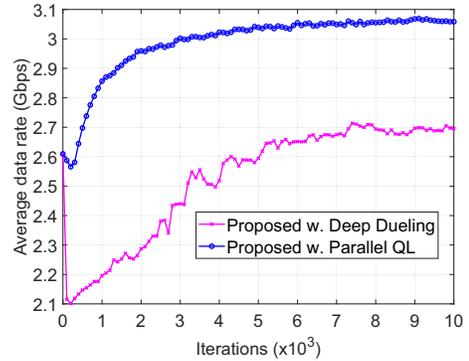}
	\caption{Convergence rates of parallel QL and deep dueling algorithms.}
	\label{fig:compareDueling}
\end{figure}

Finally, in Fig.~\ref{fig:compareDueling}, we compare the convergence rates of the proposed algorithm and the latest advance in deep reinforcement learning, i.e., deep dueling algorithm. In particular, the deep dueling reinforcement learning algorithm implements two flows of hidden layers to separately estimate the advantage and value functions~\cite{JSAC2},~\cite{JSAC}. Recent works demonstrated that the deep dueling algorithm is superior to other deep reinforcement learning algorithms, e.g., deep QL and deep double QL~\cite{JSAC2},~\cite{JSAC}. As shown in Fig.~\ref{fig:compareDueling}, our proposed solution can obtain the optimal beam association strategy within $6,000$ iterations while the deep dueling cannot converge to the optimal solution after $10,000$ iterations. The reason is that our proposed parallel QL algorithm can learn from multiple vehicles running on the road simultaneously. In contrast, the deep dueling algorithm only learns from a single vehicle at a time resulting in poor performance.

\section{Conclusion}
\label{sec:conclusion}
In this paper, we have developed an optimal beam association framework for high mobility mmWave vehicular networks, aiming to maximize the system performance in terms of average data rates, number of handovers, and disconnection probability of vehicles. The proposed parallel QL algorithm leverages the inherent feature of vehicular networks that there are usually multiple vehicles on the road. By collecting experiences/samples simultaneously from all vehicles, the algorithm converges to the optimal policy much faster than the conventional QL and even its latest deep reinforcement learning framework. Extensive simulations have proved that our proposed parallel QL algorithm can increase the average data rate by 47\% and reduce the disconnection probability by 29\% compared to the conventional solution. In addition, by learning the RSSI profiles of beams and blockages on the road, our proposed solution can achieve the performance close to that of the hypothetical scheme which requires complete environment information in advance. We also observed that the high mobility of the vehicles was actually helpful in speeding up the convergence of the algorithm to the optimal policy.

\appendices

\section{The proof of Lemma~\ref{Lem:exist}}
\label{app:reward}
First, we define a sequence of matrices as $\{A_n: n \geq 0\}$. If $\lim\limits_{n \rightarrow \infty} A_n(s'|s)=(s'|s), \forall (s,s') \in \mathcal{S} \times \mathcal{S}$, we have $\lim\limits_{n \rightarrow \infty} A_n=A$. Now, we define the Cesaro limit (denoted by C-lim)~\cite{Puterman_1994_Book} of the sequence as follows:
\begin{equation}
C\mbox{-}\!\!\lim\limits_{\!\!\!N \rightarrow \infty} = \lim\limits_{N \rightarrow \infty}\frac{A_0+ A_1+\ldots+A_n+\ldots+A_{N-1}}{N}.
\end{equation}
Thus, $A$ is the Cesaro limit (of order one) of $\{A_n: n \geq 0\}$ if
\begin{equation}
\lim\limits_{N \rightarrow \infty}\frac{1}{N}\sum_{n=0}^{N-1}A_n=A.
\end{equation}
In a short form, we have $C\mbox{-}\!\!\lim\limits_{\!\!\!N \rightarrow \infty}A_N=A$. The limiting matrix $\overline{\mathcal{L}}$ is then formulated as follows:
\begin{equation}
\overline{\mathcal{L}}=C\mbox{-}\!\!\lim\limits_{\!\!\!N \rightarrow \infty}\mathcal{L}^N.
\end{equation}
Let's denote $\overline{l}(s'|s)$ as the $(s'|s)$-th element of $\overline{\mathcal{L}}$. Thus, for each $s$ and $s'$, we have the following:
\begin{equation}
\overline{l}(s'|s) = \lim_{N \rightarrow \infty}\frac{1}{N}\sum_{n=1}^{N}l^{n-1}(s'|s),
\end{equation}
where $l^0(s'|s)$ denotes a element of an $\mathcal{S}\times\mathcal{S}$ identity matrix, and $l^{n-1}$ is a component of $\mathcal{L}^{n-1}$. Given $\mathcal{L}$ is aperiodic, $\lim_{N \rightarrow \infty}\mathcal{L}^N$ equals to $\overline{\mathcal{L}}$. Thus, the limiting matrix exists.
\section{The proof of Theorem~\ref{theo:limitexists}}
\label{appendix:recurrent}
In this proof, we first show that the underlying Markov chain is irreducible. In particular, we will prove that the learning process can move from a given state to any states after a finite number of steps. As mentioned, the system state space $\mathcal{S}$ is the combination of the RSSI level, the connected beam, and the velocity of the current vehicle. At state $s=(l, b_{n,k}, v, d)$, if the vehicle to connect to beam $b_{n',k'}$ and the RSSI level when connect to this beam is $l'$, the system moves to state $s'=(l',b_{n',k'}, v, d)$. The new RSSI level $l'$ can be any of levels in $\mathcal{R}$ as the RSSI level depends on the environmental conditions, e.g., channel conditions, and the blockage probability. In addition, the vehicle can be able to connect to all beams when it is moving on the road. When the vehicle moves out of the considered road, the system will wait for a new vehicle enters the road and move to a new state. Moreover, the velocity (i.e., speed and direction) of the vehicle is not fixed. Thus, from a given state $s$, the system can move to any other state $s' \in \mathcal{S}$ after a finite number of steps. In other words, the state space $\mathcal{S}$ (which is the combination of the RSSI level, the connected beam, and the velocity of the current vehicle) contains only one communicating class, and the underlying Markov chain is irreducible. As such, the long-term average date rate $\mathcal{R}(\pi)$ does not depend on the initial state and is well defined $\forall \pi$~\cite{CompetitiveBook}. Thus, the algorithm can converge to the optimal association policy regardless of the initial system state.

\section{The proof of Theorem~\ref{theo:convergeQ}}
\label{appendix:convergeQ}
In this proof, we show that the proposed parallel QL algorithm is ensured to converge to the optimal policy, i.e., $\mathcal{Q}_t(s,a) \rightarrow \mathcal{Q}^*(s,a)$ as $t \rightarrow \infty$. As mentioned in Section~\ref{subsec:parallelQ}, the learning processes in our proposed algorithm are serializable. Thus, the convergence proof of the parallel QL is similar to that of the QL algorithm.

The key idea of this proof is using the action-replay process (ARP) (i.e., an artificial controlled Markov decision process)~\cite{Watkins1992QLearning}. This action-relay process is defined based on the episode sequence and the learning rate. First, we denote $\{\langle s,t\rangle \}$ as the state space of the ARP~\cite{Watkins1992QLearning}. Here, $s$ is a state in the actual process, $t \geq 1$ denotes the ARP's level. In addition, the action space of the ARP is denoted as $\{a\}$ in which $a$ is a action in the actual process. Next, at state $\langle s,t\rangle$, if action $a$ is chosen, the state transition consequence and the stochastic reward of the ARP can be formulated as follows:
\begin{equation}
\mathbf{i_*}	=	\left\{	\begin{array}{ll}
\argmax_i\{t^i \le t\},	&	\mbox{if $(s,a)$ has been taken before}\\
&\mbox{decision epoch t},	\\
0,	&	\mbox{otherwise},
\end{array}	\right.
\end{equation}
where $t^i$ represents the $i^{th}$ time when performing action $a$ given state $s$. As such, $t^{i_*}$ is the last time at which action $a$ is taken at state $s$ in the real process before decision epoch $t$. The reward equals to $\mathcal{Q}_0(s,a)$ if $i_*=0$. Moreover, in this case, the action-replay process is absorbed. Otherwise, we denote the index of the decision epoch which is taken from the existing samples from the real process as follows:
\begin{equation}
\mathbf{i_e}	=	\left\{	\begin{array}{ll}
 i_*,	&	\mbox{with probability $\tau_{t^{i_*}}$},	\\
 i_*-1,	&	\mbox{with probability $(1-\tau_{t^{i_*}})\tau_{t^{i_*-1}}$},	\\
 i_*-2,	&	\mbox{with probability $(1-\tau_{t^{i_*}})(1-\tau_{t^{i_*-1}})\tau_{t^{i_*-2}}$},	\\
& \vdots\\
 0,	&	\mbox{with probability $\prod_{i=1}^{i_*}(1-\tau_{t^i})$},
\end{array}	\right.
\end{equation}
Similar as above, when $i_e=0$, the reward is $\mathcal{Q}_0(s,a)$ and the process is absorbed. Otherwise, taking $i_e \# 0$ results reward $r_{t^{i_e}}$ and a state transition to $\langle {s'}_{t^{i_e}},t^{i_e} - 1\rangle$.

Putting the above and Lemma B in~\cite{Watkins1992QLearning} together, we have $\mathcal{Q}_t(s,a) \rightarrow \mathcal{Q}^*_{ARP}(\langle s, t\rangle, a), \forall a, s,$ and $t \geq 0$, in which $\mathcal{Q}_t(s,a)$ is the optimal action values of the ARP with state $\langle s, t\rangle$ and action $a$~\cite[Lemma A]{Watkins1992QLearning}. Let's denote $r^*$ as the bound of the reward, and thus $r* \geq |r_t|, \forall t$. With loss of generality, assuming that $\mathcal{Q}_t(s,a) < \frac{r^*}{(1-\gamma)}$ with $r^* \geq 1$~\cite{Watkins1992QLearning}. Thus, with $\chi > 0$, we can find $\xi$ so that $\gamma^\xi \frac{r*}{1-\gamma} < \frac{\chi}{6}.$

By using Lemma B.4 in~\cite{Watkins1992QLearning}, the comparison of between the value of performing $a_1, \ldots, a_\xi$ in the real process, i.e., $\bar{\mathcal{Q}}(s, a_1, \ldots, a_\xi)$, with that of taking these actions in the ARP, i.e., $\bar{\mathcal{Q}}_{ARP}(\langle s,t \rangle, a_1, \ldots, a_\xi)$, is formulated as follows:
\begin{equation}
\label{eq:difference}
\begin{aligned}
&|\bar{\mathcal{Q}}_{ARP}(<s,t>, a_1, \ldots, a_\xi)-\bar{\mathcal{Q}}(s, a_1, \ldots, a_\xi)| < \\
& \frac{\chi(1-\gamma)}{6\xi r*}\frac{2\xi r*}{1-\gamma} + \frac{2\chi}{3\xi(\xi + 1)}\frac{\xi (\xi + 1)}{2} = \frac{2 \chi}{3}.
\end{aligned}
\end{equation}
Based on Lemma B.4 in~\cite{Watkins1992QLearning}, we can say that taking only $\xi$ actions results in a small different of less than $\frac{\chi}{6}$ for both the real process and the action-replay process. In addition, we can apply (\ref{eq:difference}) to an set of actions in both the action-replay and the real processes. As such, $\mathcal{Q}^*_{ARP}(\langle s, t\rangle,a) - \mathcal{Q}^*(s,a)| < \chi$. Thus, $\mathcal{Q}_t(s,a) \rightarrow \mathcal{Q}^*(s,a)$ when $t \rightarrow \infty$ with probability 1.


\end{document}